\documentclass[copyright,creativecommons]{eptcs}

\newcommand{\typeof}{1} %

\usepackage{ifthen}
\newcommand{\condinc}[2]{\ifthenelse{\equal{\typeof}{0}}{#1}{#2}}

\condinc
{
\usepackage{a4wide}
}
{
\providecommand{\event}{EXPRESS 2011}

\providecommand{\firstpage}{1}

\usepackage{breakurl}
}

\usepackage{proof}
\usepackage{subfigure}
\usepackage{epsf}
\usepackage{graphics}
\usepackage{wrapfig}
\usepackage{mathrsfs}
\usepackage{url}
\usepackage{calc}
\usepackage{amsmath,amssymb,amsfonts,mathrsfs,latexsym,stmaryrd}
\usepackage{proof}
\usepackage[all]{xy}

\newcommand{\pidill}{\ensuremath{\pi\mathtt{DILL}}}
\newcommand{\pidsll}{\ensuremath{\pi\mathtt{DSLL}}}
\newcommand{\lilo}{\ensuremath{\mathtt{LL}}}
\newcommand{\sll}{\ensuremath{\mathtt{SLL}}}

\newcommand{\tyg}[5]{#1;#2;#3\vdash #4::#5}
\newcommand{\tygcf}[4]{#1;#2\vdash #3::#4}

\newcommand{\Lone}{\mathbf{1}\mathsf{L}}
\newcommand{\Rone}{\mathbf{1}\mathsf{R}}
\newcommand{\Lten}{\otimes\mathsf{L}}
\newcommand{\Rten}{\otimes\mathsf{R}}
\newcommand{\Llin}{\multimap\mathsf{L}}
\newcommand{\Rlin}{\multimap\mathsf{R}}
\newcommand{\cut}{\mathsf{cut}}
\newcommand{\cutb}{\mathsf{cut}_!}
\newcommand{\cutd}{\mathsf{cut}_\#}
\newcommand{\cutw}{\mathsf{cut}_w}
\newcommand{\bemb}{\flat_!}
\newcommand{\bemd}{\flat_\#}
\newcommand{\Lbangb}{\mathbf{!}\mathsf{L}_{!}}
\newcommand{\Lbangd}{\mathbf{!}\mathsf{L}_{\#}}
\newcommand{\Rbang}{\mathbf{!}\mathsf{R}}
\newcommand{\Lplus}{\oplus\mathsf{L}}
\newcommand{\Rplusone}{\oplus\mathsf{R}_1}
\newcommand{\Rplustwo}{\oplus\mathsf{R}_2}
\newcommand{\Lwithone}{\with\mathsf{L}_1}
\newcommand{\Lwithtwo}{\with\mathsf{L}_2}
\newcommand{\Rwith}{\with\mathsf{R}}

\newcommand{\PLone}[2]{\mathbf{1}\mathsf{L}(#1,#2)}
\newcommand{\PRone}{\mathbf{1}\mathsf{R}}
\newcommand{\PLten}[4]{\otimes\mathsf{L}(#1,#2.#3.#4)}
\newcommand{\PRten}[2]{\otimes\mathsf{R}(#1,#2)}
\newcommand{\PLlin}[4]{\multimap\mathsf{L}(#1,#2,#3.#4)}
\newcommand{\PRlin}[2]{\multimap\mathsf{R}(#1.#2)}
\newcommand{\Pcut}[3]{\mathsf{cut}(#1,#2.#3)}
\newcommand{\Pcutb}[3]{\mathsf{cut}_!(#1,#2.#3)}
\newcommand{\Pcutd}[3]{\mathsf{cut}_\#(#1,#2.#3)}
\newcommand{\Pcutw}[3]{\mathsf{cut}_w(#1,#2.#3)}
\newcommand{\Pbemb}[3]{\flat_!(#1,#2.#3)}
\newcommand{\Pbemd}[3]{\flat_\#(#1,#2.#3)}
\newcommand{\PLbangb}[2]{\mathbf{!}\mathsf{L}_{!}(#1.#2)}
\newcommand{\PLbangd}[2]{\mathbf{!}\mathsf{L}_{\#}(#1.#2)}
\newcommand{\PRbang}[2]{\mathbf{!}\mathsf{R}(#1,#2)}
\newcommand{\PLplus}[5]{\oplus\mathsf{L}(#1,#2.#3,#4.#5)}
\newcommand{\PRplusone}[1]{\oplus\mathsf{R}_1(#1)}
\newcommand{\PRplustwo}[1]{\oplus\mathsf{R}_2(#1)}
\newcommand{\PLwithone}[3]{\with\mathsf{L}_1(#1,#2.#3)}
\newcommand{\PLwithtwo}[3]{\with\mathsf{L}_2(#1,#2.#3)}
\newcommand{\PRwith}[2]{\with\mathsf{R}(#1,#2)}
\newcommand{\lift}[1]{#1_{\Downarrow}}

\newcommand{\weip}[2]{\mathbb{W}_{#1}(#2)}
\newcommand{\wei}[1]{\mathbb{W}(#1)}
\newcommand{\bde}[1]{\mathbb{B}(#1)}
\newcommand{\dupf}[1]{\mathbb{D}(#1)}
\newcommand{\foc}[2]{\mathbb{FO}(#1,#2)}

\newcommand{\cred}{\Longrightarrow}
\newcommand{\cpred}{\longmapsto}
\newcommand{\cequ}{\equiv}
\newcommand{\cpredequ}{\hookrightarrow}
\newcommand{\reds}{\rightarrow}

\newcommand{\ptone}{\mathsf{D}}
\newcommand{\pttwo}{\mathsf{E}}
\newcommand{\ptthree}{\mathsf{F}}
\newcommand{\ptfour}{\mathsf{G}}
\newcommand{\ptfive}{\mathsf{H}}

\newcommand{\conone}{\Gamma}
\newcommand{\contwo}{\Delta}
\newcommand{\conthree}{\Theta}
\newcommand{\confour}{\Phi}
\newcommand{\confive}{\Psi}
\newcommand{\emcon}{\emptyset}

\newcommand{\procone}{P}
\newcommand{\proctwo}{Q}
\newcommand{\procthree}{R}
\newcommand{\procfour}{S}
\newcommand{\size}[1]{|#1|}
\newcommand{\bd}[1]{\mathbb{B}(#1)}

\newcommand{\tcone}{T}

\newcommand{\typone}{A}
\newcommand{\typtwo}{B}
\newcommand{\typthree}{C}

\newcommand{\cone}{x}
\newcommand{\ctwo}{y}
\newcommand{\cthree}{z}
\newcommand{\cfour}{w}

\newcommand{\tdone}{\pi}
\newcommand{\tdtwo}{\rho}

\newcommand{\unit}{\textbf{1}}

\newcommand{\tens}{\otimes}
\newcommand{\lin}{\multimap}
\newcommand{\plus}{\oplus}
\newcommand{\with}{\&}
\newcommand{\bang}{!}

\newcommand{\midd}{\; \; \mbox{\Large{$\mid$}}\;\;}

\newcommand{\insc}[2]{#1(#2)}
\newcommand{\outsc}[2]{#1\langle #2\rangle}

\newcommand{\inc}[3]{#1(#2).#3}
\newcommand{\inwc}[2]{#1(#2)}
\newcommand{\outc}[3]{#1\langle #2\rangle.#3}
\newcommand{\outwc}[2]{#1\langle #2\rangle}
\newcommand{\binc}[3]{!#1(#2).#3}
\newcommand{\para}[2]{#1\;|\;#2}
\newcommand{\rest}[2]{(\nu #1)#2}
\newcommand{\restp}[3]{(\nu #1)^{#2} #3}
\newcommand{\emproc}{0}
\newcommand{\inl}[2]{#1.\mathtt{inl};#2}
\newcommand{\inr}[2]{#1.\mathtt{inr};#2}
\newcommand{\case}[3]{#1.\mathtt{case}(#2,#3)}
\newcommand{\dotrest}[3]{\rest{#1\ldots #2}{#3}}
\newcommand{\dotpara}[2]{#1||\ldots||#2}
\newcommand{\dupserver}{\mathit{dupser}}
\newcommand{\dupclient}{\mathit{dupclient}}
\newcommand{\mulser}{\mathit{mulser}}

\newcommand{\inlm}{\mathtt{inl}}
\newcommand{\inrm}{\mathtt{inr}}
\newcommand{\casem}{\mathtt{case}}

\newcommand{\fn}[1]{\textit{fn}(#1)}
\newcommand{\subst}[3]{#1\{#2/#3\}}

\newcommand{\scon}{\equiv}
\newcommand{\aequ}{\equiv_\alpha}

\newcommand{\natone}{n}
\newcommand{\nattwo}{m}
\newcommand{\polyone}{p}
\newcommand{\polytwo}{q}
\newcommand{\NN}{\mathbb{N}}
\newcommand{\indone}{x}

\newtheorem{lemma}{Lemma}
\newtheorem{proposition}{Proposition}
\newtheorem{theorem}{Theorem}
\newtheorem{corollary}{Corollary}
\newenvironment{proof}{\begin{trivlist}
       \item[\hskip \labelsep {\bfseries Proof.}]}{\hfill $\Box$ \end{trivlist}}
\newtheorem{definition}{Definition}

\newenvironment{varitemize}
{
\begin{list}{\labelitemi}
{\setlength{\itemsep}{0.0mm}
 \setlength{\topsep}{0.0mm}
 \setlength{\parindent}{0.0mm}
 \setlength{\parskip}{0.0mm}
 \setlength{\parsep}{0.0mm}
 \setlength{\partopsep}{0.0mm}
 \setlength{\leftmargin}{15pt}
 \setlength{\labelsep}{5pt}
 \setlength{\labelwidth}{10pt}}}
{
 \end{list} 
}

\newcounter{number}

\newenvironment{varenumerate}
{\begin{list}{\arabic{number}.}
  {
   \usecounter{number}
   \setlength{\labelwidth}{4.0mm}
   \setlength{\labelsep}{2.0mm}
   \setlength{\itemindent}{0.0mm}
   \setlength{\itemsep}{0.0mm}
   \setlength{\topsep}{0.0mm}
   \setlength{\parskip}{0.0mm}
   \setlength{\parsep}{0.0mm}
   \setlength{\partopsep}{0.0mm}
  }
}
{\end{list}}

\title{Soft Session Types}
\condinc
{
\author{Ugo Dal Lago\footnote{Universit\`a di Bologna \& INRIA Sophia Antipolis, \texttt{dallago@cs.unibo.it}}
        \and
        Paolo Di Giamberardino\footnote{Laboratoire d'Informatique de Paris Nord, \texttt{digiambe@lipn.univ-paris13.fr}}}
}
{
\author{
Ugo Dal Lago
\institute{Universit\`a di Bologna \\ INRIA Sophia Antipolis}
\email{dallago@cs.unibo.it}
\and
Paolo Di Giamberardino
\institute{Laboratoire d'Informatique\\ de Paris Nord}
\email{digiambe@lipn.univ-paris13.fr}}

\def\titlerunning{Soft Session Types}
\def\authorrunning{U. Dal Lago \& P. Di Giamberardino}

 \setcounter{topnumber}{9} 
 \setcounter{bottomnumber}{9}
 \setcounter{totalnumber}{20} 
 \setcounter{dbltopnumber}{9}
}

\begin{document}
\maketitle
\begin{abstract}
\noindent
We show how systems of session types can enforce interactions 
to be bounded for all typable processes. The type system we propose 
is based on Lafont's soft linear logic and is strongly inspired by 
recent works about session types as intuitionistic linear logic formulas.
Our main result is the existence, for every typable process, of
a polynomial bound on the length of any reduction sequence starting from
it and on the size of any of its reducts.
\end{abstract}
\section{Introduction}
Session types are one of the most successful paradigms around which communication can be
disciplined in a concurrent or object-based environment. They can come in many different
flavors, depending on the underlying programming language and on the degree of flexibility
they allow when defining the structure of sessions. As an example, systems of session types for multi-party 
interaction have been recently introduced~\cite{Honda08}, while a form of higher-order session has
been shown to be definable~\cite{Mostrous07}. Recursive types, on the other hand, have been part of the standard toolset
of session type theories since their inception~\cite{Honda98}. 

The key property induced by systems of session types
is the following: if two (or more) processes can be typed with ``dual'' session types, then they can
interact with each other without ``going wrong'', i.e. avoiding situations where one party
needs some data with a certain type and the other(s) offer something of a different, incompatible
type. Sometimes, one would like to go beyond that and design a type system which guarantees stronger
properties, including quantitative ones. An example of a property that we find particularly interesting
is the following: suppose that two processes $\procone$ and $\proctwo$ interact by creating 
a session having type $\typone$ through which they communicate. Is this interaction guaranteed 
to be finite? How long would it last? Moreover, $\procone$ and $\proctwo$ could be forced
to interact with other processes in order to be able to offer $\typone$. The question could then 
become: can the global amount of interaction be kept under control? In other words, one
could be interested in \emph{proving the interaction induced by sessions to be bounded}.
This problem has been almost neglected by the research community in the area of session types, 
although it is the \emph{manifesto} of the so-called implicit computational complexity (ICC), where
one aims at giving machine-free characterizations of complexity classes based on programming
languages and logical systems.

Linear logic (\lilo\ in the following) has been introduced twenty-five years ago 
by Jean-Yves Girard~\cite{Girard87}. One of its greatest merits has been to allow a 
finer analysis of the computational content of both intuitionistic and 
classical logic. In turn, this is possible by distinguishing multiplicative
as well as additive connectives, by an involutive notion of negation, and by giving a new
status to structural rules allowing them to be applicable only to modal formulas. One
of the many consequences of this new, refined way of looking at proof theory has been
the introduction of natural characterizations of complexity classes by fragments of linear
logic. This is possible because linear logic somehow ``isolates'' complexity in the modal fragment
of the logic (which is solely responsible for the hyperexponential complexity of cut elimination
in, say intuitionistic logic), which can then be restricted so as to get exactly the
expressive power needed to capture small complexity classes. One of the simplest and most
elegant of those systems is Lafont's soft linear logic (\sll\ in the following), which has been shown to correspond to
polynomial time in the realm of classical~\cite{Lafont04}, quantum~\cite{DalLago10a} and 
higher-order concurrent computation~\cite{DalLago10b}.

Recently, Caires and Pfenning~\cite{Caires10} have shown how a system of session types can be built around
intuitionistic linear logic, by introducing \pidill, a type system for the $\pi$-calculus
where types and rules are derived from the ones of intuitionistic linear logic.
In their system, multiplicative connectives like $\tens$ and $\lin$ allow to
model sequentiality in sessions, while the additive connectives $\with$ and
$\plus$ model external and internal choice, respectively. The modal connective $\bang$, on the other
hand, allows to model a server of type $!\typone$ which can offer the functionality expressed
by $\typone$ multiple times.

In this paper, we study a restriction of \pidill, called \pidsll, which can be thought
of as being derived from \pidill\ in the same way as \sll\ is obtained from \lilo. In other
words, the operator $\bang$ behaves in \pidsll\ in the same way as in \sll. The main
result we prove about \pidsll\ is precisely about bounded interaction: whenever $\procone$
can be typed in \pidsll\ and $\procone\reds^\natone\proctwo$, then both $\natone$ and
$\size{\proctwo}$ (the size of the process $\proctwo$, to be defined later) are polynomially
related to $\size{\procone}$. This ensures an abstract but quite strong form of bounded
interaction. Another, perhaps more ``interactive'' formulation of the same result
is the following: if $\procone$ and $\proctwo$ interact via a channel of type
$\typone$, then the ``complexity'' of this interaction is bounded by a polynomial
on $\size{\procone}+\size{\proctwo}$, whose degree only depends on $\typone$.

We see this paper as the first successful attempt to bring techniques
from implicit computational complexity into the realm of session types. Although proving
bounded interaction has been technically nontrivial, due to the peculiarities of
the $\pi$-calculus, we think the main contribution of this work lies in showing that
bounded termination can be enforced by a natural adaptation of known systems of
session types.
\condinc{}{An extended version with more details is available~\cite{DalLagoDiGiambe2011ev}.}

\section{\pidill, an Informal Account}\label{sect:pidillia}
In this section, we will outline the main properties of \pidill, a session type system recently introduced
by Caires and Pfenning~\cite{Caires10,Caires11}. For more information, please consult the two cited papers.

In \pidill, session types are nothing more than formulas of (propositional) intuitionistic linear
logic without atoms but with (multiplicative) constants:
$$
\typone::=\unit\midd\typone\tens\typone\midd\typone\lin\typone\midd\typone\plus\typone\midd\typone\with\typone\midd\;\bang\typone.
$$
These types are assigned to channels (names) by a formal system deriving judgments in the
form
$$
\tygcf{\conone}{\contwo}{\procone}{\cone:\typone},
$$
where $\conone$ and $\contwo$ are contexts assigning types to channels, and $\procone$ is a process of
the name-passing $\pi$ calculus. The judgment above can be read as follows: the process $\procone$ acts
on the channel $\cone$ according to the session type $\typone$ \emph{whenever} composed with processes
behaving according to $\conone$ and $\contwo$ (each on a specific channel).
Informally, the various constructions on session types can be explained as follows:
\begin{varitemize}
\item
  $\unit$ is the type of an empty session channel. A process offering to communicate via a session channel 
  typed this way simply synchronizes with another process through it without exchanging anything. This is 
  meant to be an abstraction for all ground session types, e.g. natural
  numbers, lists, etc. In linear logic, this is the unit for $\tens$.
\item
  $\typone\tens\typtwo$ is the type of a session channel $\cone$ through which a message carrying another       
  channel with type $\typone$ is sent. After performing this action, the underlying process behaves according 
  to $\typtwo$ on the \emph{same} channel $\cone$.
\item
  $\typone\lin\typtwo$ is the adjoint to $\typone\tens\typtwo$: on a channel with this type, a process 
  communicate by first performing an input and receiving a channel with type $\typone$, then acting according 
  to $\typtwo$, again on $\cone$.
\item
  $\typone\plus\typtwo$ is the type of a channel on which a process either sends a special message 
  $\inlm$ and performs according to $\typone$ or sends a special message $\inrm$ and performs according to
   $\typtwo$.
\item
  The type $\typone\with\typtwo$ can be assigned to a channel $\cone$ on which the underlying process offers 
  the possibility of choosing between proceeding according to $\typone$ or to $\typtwo$, both on $\cone$. 
  So, in a sense, $\with$ models external choice.
\item
  Finally, the type $\bang\typone$ is attributed to a channel $\cone$ only if a process can be 
  replicated by receiving a channel
  $\ctwo$ through $\cone$, then behaving on $\ctwo$ according to $\typone$.
\end{varitemize}
\newcommand{\dill}{\ensuremath{\texttt{DILL}}}
The assignments in $\conone$ and $\contwo$ are of two different natures: 
\begin{varitemize}
\item
  An assignment of a type $\typone$ to a channel $\cone$ in $\contwo$ signals the need by $\procone$
  of a process offering a session of type $\typone$ on the channel $\cone$; for this reason, $\contwo$
  is called the \emph{linear context};
\item
  An assignment of a type $\typone$ to a channel $\cone$ in $\conone$, on the other hand, represents
  the need by $\procone$ of a process offering a session of type $\bang\typone$ on the channel $\cone$;
  thus, $\conone$ is the \emph{exponential context}. 
\end{varitemize}
Typing rules \pidill\ are very similar to the ones of \dill, itself one of the many possible formulations of linear logic
as a sequent calculus. In particular, there are two cut rules, each corresponding to a different portion of the context:
$$
\begin{array}{ccc}
\infer
 {\tygcf{\conone}{\contwo_1,\contwo_2}{\rest{\cone}{(\para{\procone}{\proctwo})}}{\tcone}}
 {\tygcf{\conone}{\contwo_1}{\procone}{\cone:A} & \tygcf{\conone}{\contwo_2, \cone:A}{\proctwo}{\tcone}}
&
&
\infer
 {\tygcf{\conone}{\contwo}{\rest{\cone}{(\para{\binc{\cone}{\ctwo}{\procone}}{\proctwo})}}{\tcone}}
 {\tygcf{\conone}{\emcon}{\procone}{\ctwo:A} & \tygcf{\conone,\cone:A}{\contwo}{\proctwo}{\tcone}}
\end{array}
$$
Please observe how cutting a process $\procone$ against an assumption in the exponential context
requires to ``wrap'' $\procone$ inside a replicated input: this allows to \emph{turn $\procone$ into
a server}.

\newcommand{\natnum}{\mathbf{N}}
\newcommand{\mpthree}{\mathbf{S}}
In order to illustrate the intuitions above, we now give an example. Suppose that a process $\procone$ models a service which
acts on $\cone$ as follows: it receives two natural numbers, to be interpreted as the number and secret code of a credit
card and, if they correspond to a valid account, returns an MP3 file and a receipt code to the client. Otherwise, the
session terminates. To do so, $\procone$ needs to interact with another service (e.g. a banking service) $\proctwo$ through a
channel $\ctwo$. The banking service, among others, provides a way to verify whether a given number and code correspond to 
a valid credit card. In \pidill, the process $\procone$ would receive the type
$$
\tygcf{\emcon}{\ctwo:(\natnum\lin\unit\plus\unit)\with\typone}{\procone}{\cone:\natnum\lin\natnum\lin(\mpthree\tens\natnum)\plus\unit},
$$
where $\natnum$ and $\mpthree$ are pseudo-types for natural numbers and MP3s, respectively. $\typone$ is the type of all the other
functionalities $\proctwo$ provides. As an example, $\procone$ could be the following process:
\newcommand{\cnumbone}{\mathit{nm}_1}
\newcommand{\ccodeone}{\mathit{cd}_1}
\newcommand{\cnumbtwo}{\mathit{nm}_2}
\newcommand{\ccodetwo}{\mathit{cd}_2}
\newcommand{\cmpthree}{\mathit{mp}}
\newcommand{\crecp}{\mathit{rp}}
\begin{align*}
&
\inc{\cone}{\cnumbone}{
\inc{\cone}{\ccodeone}{
\inl{\ctwo}{\\
&\qquad
\rest{\cnumbtwo}{
\outc{\ctwo}{\cnumbtwo}{
\rest{\ccodetwo}{
\outc{\ctwo}{\ccodetwo}{\\
&\qquad\qquad
\case{\ctwo}
  { 
    \inl{\cone}{
    \rest{\cmpthree}{
    \outc{\cone}{\cmpthree}{
    \rest{\crecp}{
    \outwc{\cone}{\crecp}
    }}}}
  }
  {
    \inr{\cone}{\emproc}
  }
}}}}}}}
\end{align*}
Observe how the credit card number and secret code forwarded to $\proctwo$ are not the ones sent by the client: the flow of information
happening inside a process is abstracted away in \pidill. Similarly, one can write a process $\proctwo$ and assign it a type as follows:
$\tygcf{\emcon}{\emcon}{\proctwo}{\ctwo:(\natnum\lin\unit\plus\unit)\with\typone}$. 
Putting the two derivations together, we obtain
$\tygcf{\emcon}{\emcon}{\rest{\cone}{(\para{\procone}{\proctwo})}}{\cone:\natnum\lin\natnum\lin(\mpthree\tens\natnum)\plus\unit}$.

Let us now make an observation which will probably be appreciated by the reader familiar with linear logic. The processes $\procone$
and $\proctwo$ can be typed in \pidill\ without the use of any exponential rule, nor of cut. What allows to type the
parallel composition $\rest{\cone}{(\para{\procone}{\proctwo})}$, on the other hand, is precisely the cut rule. The interaction between 
$\procone$ and $\proctwo$ corresponds to the elimination of that cut. Since there isn't any exponential around, this process must be
finite, since the size of the underlying process shrinks at every single reduction step. From a process-algebraic point of view,
on the other hand, the finiteness of the interaction is an immediate consequence of the absence of any replication in $\procone$
and $\proctwo$.

The banking service $\proctwo$ can only serve one single session and would vanish at the end of it. To make it into a \emph{persistent
server} offering the same kind of session to possibly many different clients, $\proctwo$ must be put into a replication,
obtaining $\procthree=\binc{\cthree}{\ctwo}{\proctwo}$. In $\procthree$, the channel $\cthree$ can be given type 
$\bang((\natnum\lin\unit\plus\unit)\with\typone)$  in the empty context. The process $\procone$ should 
be somehow adapted to be able to interact with $\procthree$: before performing the
 two outputs on $\ctwo$, it's necessary to ``spawn'' $\procthree$ by performing an output on $\cthree$ and passing $\ctwo$ to it. This way
we obtain a process $\procfour$ such that
$$
\tygcf{\emcon}{\cthree:\bang((\natnum\lin\unit\plus\unit)\with\typone)}{\procfour}{\cone:\natnum\lin\natnum\lin(\mpthree\tens\natnum)\plus\unit},
$$
and the composition $\rest{\cthree}{(\para{\procfour}{\procthree})}$ can be given the same type as 
$\rest{\cone}{(\para{\procone}{\proctwo})}$. Of course, $\procfour$ could have used the channel $\cthree$ more than once, initiating
different sessions. This is meant to model a situation in which the same client interacts with the same server by creating more
than one session with the same type, itself done by performing \emph{more than one output} on the same channel. Of course, servers can themselves 
depend on other servers. And these dependencies are naturally modeled by the exponential modality of linear logic.

\section{On Bounded Interaction}\label{sect:onbint}
In \pidill, the possibility of modeling persistent servers which in turn depend on other servers makes it possible to type processes 
which exhibit a very complex and combinatorially heavy interactive behavior.

Consider the following processes, the first one parameterized on any $i\in\NN$:
\begin{align*}
\dupserver_i&\doteq\;\binc{\cone_i}{\ctwo}{\rest{\cthree}{\outc{\cone_{i+1}}{\cthree}{\rest{\cfour}{\outc{\cone_{i+1}}{\cfour}}}}};\\
\dupclient&\doteq\rest{\ctwo}{\outsc{\cone_{0}}{\ctwo}}.
\end{align*}
In \pidill, these processes can be typed as follows:
\begin{align*}
\tygcf{\emcon}{\cone_{i+1}:\bang{\unit}}{&\dupserver_i}{\cone_{i}:\bang{\unit}};\\
\tygcf{\emcon}{\cone_{0}:\bang{\unit}}{&\dupclient}{\cthree:\unit}.
\end{align*}
Then, for every $n\in\NN$ one can type the parallel composition
\condinc
{
$$\mulser_{n+1}\doteq\dotrest{\cone_{1}}{\cone_{n}}{(\dotpara{\dupserver_n}{\dupserver_0})}$$
}
{
$\mulser_{n+1}\doteq\dotrest{\cone_{1}}{\cone_{n}}{(\dotpara{\dupserver_n}{\dupserver_0})}$
}
as follows
$$
\tygcf{\emcon}{\cone_{n}:\bang{\unit}}{\mulser_n}{\cone_{0}:\bang{\unit}}.
$$
Informally, $\mulser_{n}$ is a persistent server which offers a session  type $\unit$ on a channel
$\cone_0$, provided a server with the same functionality is available on $\cone_n$. The
process $\mulser_{n}$ is the parallel composition of $n$ servers in
the form $\dupserver_i$, each spawning two different sessions provided
by $\dupserver_{i+1}$ on the same channel $\cone_{i+1}$.

The process $\mulser_n$ cannot be further reduced. But notice that, once
$\mulser_n$ and $\dupclient$ are composed, the following exponential blowup is bound to happen: 
\begin{align*}
  \rest{\cone_0}{(\para{\mulser_n}{\dupclient})}&\equiv\dotrest{\cone_{0}}{\cone_{n}}{(\para{\dotpara{\dupserver_n}{\dupserver_0}}{\dupclient})}\\
      &\reds\dotrest{\cone_{0}}{\cone_{n}}{(\para{\dotpara{\dupserver_n}{\dupserver_1}}{\procone_1})}\\
      &\reds^2\dotrest{\cone_{1}}{\cone_{n}}{(\para{\dotpara{\dupserver_n}{\dupserver_2}}{\para{\procone_2}{\procone_2}})}\\
      &\reds^4\dotrest{\cone_{2}}{\cone_{n}}{(\para{\dotpara{\dupserver_n}{\dupserver_3}}{\underbrace{\dotpara{\procone_3}{\procone_3}}_{\mbox{$4$ times}}})}\\
      &\reds^{*}\rest{\cone_n}{(\para{\dupserver_n}{\underbrace{\dotpara{\procone_n}{\procone_n}}_{\mbox{$2^n$ times}}})}\\
      &\reds^{2^n}\emproc.
\end{align*}
Here, for every $i\in\NN$ the process $\procone_i$ is simply $\rest{\ctwo}{\outc{\cone_{i}}{\ctwo}{\rest{\cthree}{\outwc{\cone_{i}}{\cthree}}}}$.
Notice that \emph{both} the number or reduction steps \emph{and} the size of intermediate
processes are exponential in $n$, while the size of the initial process is linear in $n$.
This is a perfectly legal process in \pidill. Moreover the type $\bang\unit$ of the channel
$\cone_0$ through which $\dupclient$ and $\mulser_n$ communicate does not contain any information
about the ``complexity'' of the interaction: it is the same for every $n$.

The deep reasons why this phenomenon can happen lie in the very general (and ``generous'') 
rules governing the behavior of the exponential modality $\bang$ in linear logic. It is this generality
that allows the embedding of propositional intuitionistic logic into linear logic.
Since the complexity of normalization for the former~\cite{Statman79,Mairson92} is nonelementary, the 
exponential blowup described above is not a surprise.

It would be desirable, on the other hand, to be sure that the interaction caused by any process $\procone$ 
is bounded: whenever $\procone\reds^\natone\proctwo$, then there's a \emph{reasonably low}
upper bound to both $\natone$ and $\size{\proctwo}$. This is precisely what we achieve by restricting
\pidill\ into \pidsll.

\section{\pidsll: Syntax and Main Properties}

In this section, the syntax of \pidsll\ will be introduced.
Moreover, some basic operational properties will be given.

\subsection{The Process Algebra}
\pidsll\ is a type system for a fairly standard $\pi$-calculus, exactly the one on top of which
\pidill\ is defined:
\begin{definition}[Processes]\label{def:scon}
Given an infinite set of \emph{names} or \emph{channels} $\cone,\ctwo,\cthree,\ldots$, the set of
\emph{processes} is defined as follows: 
$$
\procone::=\emproc\midd\para{\procone}{\proctwo}\midd\rest{\cone}{\procone}\midd\inc{\cone}{\ctwo}{\procone}\midd
   \outc{\cone}{\ctwo}{\procone}\midd\binc{\cone}{\ctwo}{\procone}\midd\inl{\cone}{\procone}\midd\inr{\cone}{\procone}
   \midd\case{\cone}{\procone}{\proctwo}
$$
\end{definition}
The only non-standard constructs are the last three, which allow to define a choice mechanism: 
the process $\case{\cone}{\procone}{\proctwo}$ can evolve as $\procone$ or as $\proctwo$ \emph{after}
having received a signal in the form $\inlm$ o $\inrm$ through $\cone$. Processes sending such
a signal through the channel $\cone$, then continuing like $\procone$ are, respectively,
$\inl{\cone}{\procone}$ and $\inr{\cone}{\procone}$. The set of names occurring free in the process
$\procone$ (hereby denoted $\fn{\procone}$) is defined as usual. The same holds for the capture
avoiding substitution of a name $\cone$ for $\ctwo$ in a process $\procone$ (denoted $\subst{\procone}{\cone}{\ctwo}$),
and for $\alpha$-equivalence between processes (denoted $\aequ$).

Structural congruence is an equivalence relation identifying those processes which are syntactically
different but can be considered equal for very simple structural reasons:
\begin{definition}[Structural Congruence]
The relation $\scon$, called \emph{structural congruence},  
is the least congruence on processes satisfying the following seven axioms:
\begin{align*}
\procone&\scon\proctwo\quad\mbox{whenever $\procone\aequ\proctwo$}; & \rest{\cone}{\emproc}&\scon\emproc;\\
\para{\procone}{\emproc}&\scon\procone; & \rest{\cone}{\rest{\ctwo}{\procone}}&\scon\rest{\ctwo}{\rest{\cone}{\procone}};\\
\para{\procone}{\proctwo}&\scon\para{\proctwo}{\procone}; & \para{(\rest{\cone}{\procone})}{\proctwo}&\scon\rest{\cone}{(\para{\procone}{\proctwo}})\quad\mbox{whenever $\cone\notin\fn{\proctwo}$};\\
\para{\procone}{(\para{\proctwo}{\procthree})}&\scon\para{(\para{\procone}{\proctwo})}{\procthree}. & &
\end{align*}
\end{definition}
Formal systems for reduction and labelled semantics can be defined in a standard way. We refer the
reader to~\cite{Caires10} for more details.

A quantitative attribute of processes which is delicate to model in process algebras is their
\emph{size}: how can we measure the size of a process? In particular, it is not straightforward
to define a measure which both reflects the ``number of symbols'' in the process and is invariant
under structural congruence (this way facilitating all proofs). A good compromise is the following:
\begin{definition}[Process Size]
The \emph{size} $\size{\procone}$ of a process $\procone$ is defined by induction on the structure of 
$\procone$ as follows:
\begin{align*}
\size{\emproc}&=0; & \size{\inc{\cone}{\ctwo}{\procone}}&=\size{\procone}+1; & \size{\inl{\cone}{\procone}}&=\size{\procone}+1; \\
\size{\para{\procone}{\proctwo}}&=\size{\procone}+\size{\proctwo}; & \size{\outc{\cone}{\ctwo}{\procone}}&=\size{\procone}+1;
   & \size{\inr{\cone}{\procone}}&=\size{\procone}+1;\\
\size{\rest{\cone}{\procone}}&=\size{\procone}; & \size{\binc{\cone}{\ctwo}{\procone}}&=\size{\procone}+1;
   & \size{\case{\cone}{\procone}{\proctwo}}&=\size{\procone}+\size{\proctwo}+1.
\end{align*}
\end{definition}
According to the definition above, the empty process $\emproc$ has null size, while restriction
does not increase the size of the underlying process. This allows for a definition of size which 
remains invariant under structural congruence. The price to pay is the following: the ``number
of symbols'' of a process $\procone$ can be arbitrarily bigger than $\size{\procone}$ (e.g.
for every $\natone\in\NN$, $\size{\restp{\cone}{\natone}{\procone}}=\size{\procone}$). However, we have 
the following:
\begin{lemma}
   For every $\procone,\proctwo$, $\size{\procone}=\size{\proctwo}$ whenever
   $\procone\scon\proctwo$.
   Moreover, there is a polynomial $\polyone$ such that for every $\procone$,
   there is $\proctwo$ with $\procone\scon\proctwo$ and the number
   of symbols in $\proctwo$ is at most $\polyone(\size{\proctwo})$.
\end{lemma}
\condinc{
\begin{proof}
The fact $\procone\scon\proctwo$ implies $\size{\procone}=\size{\proctwo}$ can be proved
by a simple inspection of Definition~\ref{def:scon}. The second part of the lemma
can be proved by induction on $\procone$ once $\polyone$ is fixed as $\procone(x)=x^2$.
\end{proof}}{}

\subsection{The Type System}
The language of types of \pidsll\ is exactly the same as the one of \pidill,
and the interpretation of type constructs does not change (see Section~\ref{sect:pidillia} for some informal
details). Typing judgments and typing rules, however, are significantly different, in particular, in the
treatment of the exponential connective $\bang$. 

Typing judgments become syntactical expressions in the form
$$
\tyg{\conone}{\contwo}{\conthree}{\procone}{\cone:\typone}.
$$
First of all, observe how the context is divided into \emph{three} chunks now:
$\conone$ and $\contwo$ have to be interpreted as exponential contexts, while $\conthree$
is the usual linear context from \pidill. The necessity of having \emph{two} exponential
contexts is a consequence of the finer, less canonical exponential discipline of \sll\ compared to the
one of \lilo. We use the following terminology: $\conone$ is said to be the \emph{auxiliary}
context, while $\contwo$ is the \emph{multiplexor} context.

Typing rules are in Figure~\ref{fig:pidslltr}.
\begin{figure}
\begin{center}
\fbox{\begin{minipage}[c]{.99\textwidth}
$$
\begin{array}{ccccc}
\infer[\Lone]
  {\tyg{\conone}{\contwo}{\conthree,\cone:\unit}{\procone}{\tcone}}
  {\tyg{\conone}{\contwo}{\conthree}{\procone}{\tcone}}
&
\hspace{10pt}
&
\infer[\Rone]
  {\tyg{\conone}{\contwo}{\emcon}{\emproc}{\cone:\unit}}
  {}
\end{array}
$$
\vspace{2pt}
$$
\begin{array}{ccc}
\infer[\Lten]
 {\tyg{\conone}{\contwo}{\conthree, \cone : A \tens B}{\inc{\cone}{\ctwo}{\procone}}{\tcone}}
 {\tyg{\conone}{\contwo}{\conthree, \ctwo: A, \cone: B}{\procone}{\tcone}}
&
\hspace{10pt}
&
\infer[\Rten]
 {\tyg{\conone_1,\conone_2}{\contwo}{\conthree_1,\conthree_2}{\rest{\ctwo}{\outc{\cone}{\ctwo}{(\para{\procone}{\proctwo})}}}{\cone: A \tens B}}
 {\tyg{\conone_1}{\contwo}{\conthree_1}{\procone}{\ctwo:A} & \tyg{\conone_2}{\contwo}{\conthree_2}{\proctwo}{\cone:B}}
\end{array}
$$
\vspace{2pt}
$$
\begin{array}{ccc}
\infer[\Llin]
 {\tyg{\conone_1,\conone_2}{\contwo}{\conthree_1,\conthree_2, \cone: A \lin B}
   {\rest{\ctwo}{\outc{\cone}{\ctwo}{(\para{\procone}{\proctwo})}}}{\tcone}}
 {\tyg{\conone_1}{\contwo}{\conthree_1,\ctwo: A}{\procone}{\tcone} & \tyg{\conone_2}{\contwo}{\conthree_2, \cone:B}{\proctwo}{\tcone}}
&
\hspace{-5pt}
&
\infer[\Rlin]
 {\tyg{\conone}{\contwo}{\conthree}{\inc{\cone}{\ctwo}{\procone}}{\cone : A \lin B}}
 {\tyg{\conone}{\contwo}{\conthree, \ctwo: A }{\procone}{\cone: B}}
\end{array}
$$
\vspace{2pt}
$$
\begin{array}{ccc}
\infer[\Lplus]
 {\tyg{\conone}{\contwo}{\cone:\typone\plus\typtwo,\conthree}{\case{\ctwo}{\procone}{\proctwo}}{\tcone}}
 {
   \tyg{\conone}{\contwo}{\conthree,\cone:\typone}{\procone}{\tcone} &
   \tyg{\conone}{\contwo}{\conthree,\cone:\typtwo}{\procone}{\tcone}
 }
&
\hspace{10pt}
&
\infer[\Rplusone]
  {\tyg{\conone}{\contwo}{\conthree}{\inl{\cone}{\procone}}{\cone:\typone\plus\typtwo}}
  {\tyg{\conone}{\contwo}{\conthree}{\procone}{\cone:\typone}} 
\end{array}
$$
\vspace{2pt}
$$
\begin{array}{ccc}
\infer[\Rplustwo]
  {\tyg{\conone}{\contwo}{\conthree}{\inr{\cone}{\procone}}{\cone:\typone\plus\typtwo}}
  {\tyg{\conone}{\contwo}{\conthree}{\procone}{\cone:\typtwo}} 
&
\hspace{10pt}
&
\infer[\Lwithone]
  {\tyg{\conone}{\contwo}{\conthree,\cone:\typone\with\typtwo}{\inl{\cone}{\procone}}{\tcone}}
  {\tyg{\conone}{\contwo}{\conthree,\cone:\typone}{\procone}{\tcone}} 
\end{array}
$$
\vspace{2pt}
$$
\begin{array}{ccc}
\infer[\Lwithtwo]
  {\tyg{\conone}{\contwo}{\conthree,\cone:\typone\with\typtwo}{\inr{\cone}{\procone}}{\tcone}}
  {\tyg{\conone}{\contwo}{\conthree,\cone:\typtwo}{\procone}{\tcone}} 
&
\hspace{10pt}
&
\infer[\Rwith]
 {\tyg{\conone}{\contwo}{\conthree}{\case{\ctwo}{\procone}{\proctwo}}{\cone:\typone\with\typtwo}}
 {
   \tyg{\conone}{\contwo}{\conthree}{\procone}{\cone:\typone} &
   \tyg{\conone}{\contwo}{\conthree}{\procone}{\cone:\typtwo}
 }
\end{array}
$$
\vspace{2pt}
$$
\begin{array}{ccc}
\infer[\bemd]
 {\tyg{\conone}{\contwo,\cone:A}{\conthree}{\rest{\ctwo}{\outc{\cone}{\ctwo}{\procone}}}{\tcone}}
 {\tyg{\conone}{\contwo,\cone:A}{\conthree,\ctwo: A}{\procone}{\tcone}}
&
\hspace{10pt}
&
\infer[\bemb]
 {\tyg{\conone,\cone:A}{\contwo}{\conthree}{\rest{\ctwo}{\outc{\cone}{\ctwo}{\procone}}}{\tcone}}
 {\tyg{\conone}{\contwo}{\conthree,\ctwo: A}{\procone}{\tcone}}
\end{array}
$$
\vspace{2pt}
$$
\begin{array}{ccccc}
\infer[\Lbangd]
{\tyg{\conone}{\contwo}{\conthree, \cone: \bang A}{\procone}{\tcone}}
{\tyg{\conone}{\contwo, \cone : A}{\conthree}{\procone}{\tcone}}
&
\hspace{10pt}
&
\infer[\Lbangb]
{\tyg{\conone}{\contwo}{\conthree, \cone: \bang A}{\procone}{\tcone}}
{\tyg{\conone, \cone : A}{\contwo}{\conthree}{\procone}{\tcone}}
&
\hspace{10pt}
&
\infer[\Rbang]
{\tyg{\emcon}{\contwo}{\bang\conone}{\binc{\cone}{\ctwo}{\proctwo}}{\cone: \bang A}} 
{\tyg{\conone}{\emcon}{\emcon}{\proctwo}{\ctwo:A}}
\end{array}
$$
\vspace{2pt}
\condinc{
$$
\infer[\cut]
 {\tyg{\conone_1,\conone_2}{\contwo}{\conthree_1,\conthree_2}{\rest{\cone}{(\para{\procone}{\proctwo})}}{\tcone}}
 {\tyg{\conone_1}{\contwo}{\conthree_1}{\procone}{\cone:A} & \tyg{\conone_2}{\contwo}{\conthree_2, \cone:A}{\proctwo}{\tcone}}
$$
\vspace{2pt}
$$
\infer[\cutd]
 {\tyg{\conone}{\contwo}{\conthree}{\rest{\cone}{(\para{\binc{\cone}{\ctwo}{\procone}}{\proctwo})}}{\tcone}}
 {\tyg{\contwo}{\emcon}{\emcon}{\procone}{\ctwo:A} & \tyg{\conone}{\contwo,\cone:A}{\conthree}{\proctwo}{\tcone}}
$$
}
{
$$
\begin{array}{ccc}
\infer[\cut]
 {\tyg{\conone_1,\conone_2}{\contwo}{\conthree_1,\conthree_2}{\rest{\cone}{(\para{\procone}{\proctwo})}}{\tcone}}
 {\tyg{\conone_1}{\contwo}{\conthree_1}{\procone}{\cone:A} & \tyg{\conone_2}{\contwo}{\conthree_2, \cone:A}{\proctwo}{\tcone}}
&
\hspace{-3pt}
&
\infer[\cutd]
 {\tyg{\conone}{\contwo}{\conthree}{\rest{\cone}{(\para{\binc{\cone}{\ctwo}{\procone}}{\proctwo})}}{\tcone}}
 {\tyg{\contwo}{\emcon}{\emcon}{\procone}{\ctwo:A} & \tyg{\conone}{\contwo,\cone:A}{\conthree}{\proctwo}{\tcone}}
\end{array}
$$
}
\vspace{2pt}
$$
\infer[\cutb]
 {\tyg{\conone_1,\conone_2}{\contwo}{\conthree}{\rest{\cone}{(\para{\binc{\cone}{\ctwo}{\procone}}{\proctwo})}}{\tcone}}
 {\tyg{\conone_1}{\emcon}{\emcon}{\procone}{\ctwo:A} & \tyg{\conone_2, \cone:A}{\contwo}{\conthree}{\proctwo}{\tcone}}
$$
\end{minipage}}
\end{center}
\caption{Typing rules for \pidsll.}\label{fig:pidslltr} 
\end{figure}
The rules governing the typing constant $\unit$, the multiplicatives ($\tens$ and $\lin$) and the
additives ($\plus$ and $\with$) are exact analogues of the ones from \pidill. The only differences
come from the presence of two exponential contexts: in binary multiplicative
rules ($\Rten$ and  $\Llin$) the auxiliary context is treated multiplicatively, 
while the multiplexor context is treated additively,
as in \pidill\footnote{The reader familiar with linear logic and proof nets will recognize 
in the different treatment of the auxiliary and multiplexor contexts, one of the basic 
principles of $\sll$: \emph{contraction is forbidden on the auxiliary doors of exponential boxes}. 
The channel names contained in the auxiliary context correspond to the auxiliary doors of 
exponential boxes, so we treat them multiplicatively. The contraction effect induced by 
the additive treatment of the channel names in the multiplexor context corresponds to 
the multiplexing rule of $\sll$. }. Now, consider the rules governing the exponential 
connective $\bang$, which are $\bemb$, $\bemd$, $\Lbangb$, $\Lbangd$ and $\Rbang$:
\begin{varitemize}
\item
  The rules $\bemb$ and $\bemd$ both allow to spawn a server. This corresponds to turning
  an assumption $\cone:A$ in the linear context into one $\ctwo:A$ in one of the exponential
  contexts; in $\bemd$, $\cone:A$ could be already present in the multiplexor context, while
  in $\bemb$ this cannot happen;
\item
  The rules $\Lbangb$ and $\Lbangd$ lift an assumption in the exponential contexts to
  the linear context; this requires changing its type from $\typone$ to $\bang\typone$;
\item
  The rule $\Rbang$ allows to turn an ordinary process into a server, by packaging it into
  a replicated input and modifying its type.
\end{varitemize}
Finally there are \emph{three} cut rules in the system, namely $\cut$, $\cutb$ and $\cutd$:
\begin{varitemize}
\item
  $\cut$ is the usual linear cut rule, i.e. the natural generalization of the one from \pidill.
\item
  $\cutb$ and $\cutd$ allow to eliminate an assumption in one of the the two exponential contexts. In both
  cases, the process which allows to do that must be typable with empty linear and multiplexor
  contexts.
\end{varitemize}

\subsection{Back to Our Example}\label{sect:btoe}
Let us now reconsider the example processes introduced in Section~\ref{sect:onbint}. The basic
building block over which everything is built was the process
$\dupserver_i=\binc{\cone_i}{\ctwo}{\rest{\cthree}{\outc{\cone_{i+1}}{\cthree}{\rest{\cfour}{\outc{\cone_{i+1}}{\cfour}}}}}$.
We claim that for every $i$, the process $\dupserver_i$ is \emph{not} typable in \pidsll. To understand why, observe
that the only way to type a replicated input like $\dupserver_i$ is by the typing
rule $\Rbang$, and that its premise requires the body of the replicated input to be typable with empty linear and multiplexor
contexts. A quick inspection on the typing rules reveals that every name in the \emph{auxiliary} context
occurs (free) exactly once in the underlying process (provided we count two occurrences in the
branches of a $\casem$ as just \emph{a single} occurrence). However, the name $\cone_{i+1}$ appears
\emph{twice} in the body of $\dupserver_i$.

A slight variation on the example above, on the other hand, \emph{can} be typed in \pidsll, but this
requires changing its type. 
\condinc
{}
{See~\cite{DalLagoDiGiambe2011ev} for more details.}

\subsection{Subject Reduction}
A basic property most type systems for functional languages satisfy is subject reduction: typing is preserved along
reduction. For processes, this is often true for internal reduction: if $\procone\reds\proctwo$ and $\vdash\procone:\typone$,
then $\vdash\proctwo:\typone$. In this section, a subject reduction result for \pidsll\ will be given and some ideas 
on the underlying proof will be described. Some concepts outlined here will become necessary ingredients in the
proof of bounded interaction, to be done in Section~\ref{sect:boundint} below. Subject reduction is proved by
closely following the path traced by Caires and Pfenning; as a consequence, we proceed quite quickly, concentrating our attention on
the differences with their proof.

\newcommand{\pttotd}[1]{\widehat{#1}}
When proving subject reduction, one constantly work with type derivations. This is 
particularly true here, where (internal) reduction corresponds to the cut-elimination process. A linear
notation for proofs in the form of \emph{proof terms} can be easily defined, allowing for more compact 
descriptions. As an example, a proof in the form
$$
\infer[\cut]
 {\tyg{\conone_1,\conone_2}{\contwo}{\conthree_1,\conthree_2}{\rest{\cone}{(\para{\procone}{\proctwo})}}{\tcone}}
 {\tdone:\tyg{\conone_1}{\contwo}{\conthree_1}{\procone}{\cone:A} & \tdtwo:\tyg{\conone_2}{\contwo}{\conthree_2, \cone:A}{\proctwo}{\tcone}}
$$
corresponds to the proof term $\Pcut{\ptone}{\cone}{\pttwo}$, where $\ptone$ is the proof term
for $\tdone$ and $\pttwo$ is the proof term for $\tdtwo$. If $\ptone$ is a proof term corresponding
to a type derivation for the process $\procone$, we write $\pttotd{\ptone}=\procone$. From now
on, proof terms will often take the place of processes: $\tyg{\conone}{\contwo}{\conthree}{\ptone}{\tcone}$ 
stands for the existence of a type derivation $\ptone$ with conclusion
$\tyg{\conone}{\contwo}{\conthree}{\pttotd{\ptone}}{\tcone}$. A proof term $\ptone$ is said
to be normal if it does not contain any instances of cut rules.
\condinc{In Figure~\ref{fig:typextr}  we show in detail how processes are associated with proof terms.
\begin{figure}
\begin{center}
  \begin{tabular}[t]{lll}
    $\PLone{\cone}{\ptone}$      & $\rightsquigarrow$ & $\widehat{\ptone}^{z}$\\
    $\PRone$			    & $\rightsquigarrow$ &  $0$\\
    $\PLten{\cone}{\ctwo}{\cthree}{\pttwo}$ & $\rightsquigarrow$ &    $\inc{\cone}{\ctwo}{\widehat{\pttwo}^{z}}$\\ 
    $\PRten{\ptone}{\pttwo}$  & $\rightsquigarrow$ &   $\rest{\ctwo}{\outc{\cone}{\ctwo}{(\para{\widehat{\ptone}^{y}}{\widehat{\pttwo}^{x}})}}$ \\
    $\PLlin{\cone}{\ptone}{\ctwo}{\pttwo}$ &  & $\rest{\ctwo}{\outc{\cone}{\ctwo}{(\para{\widehat{\ptone}^{y}}{\widehat{\pttwo}^{z}})}}$ \\
    $\PRlin{\cone}{\ptone}$ & $\rightsquigarrow$ &   $\inc{\cone}{\ctwo}{\widehat{\pttwo}^{x}}$  \\
    $\Pcut{\ptone}{\cone}{\pttwo}$ & $\rightsquigarrow$ &  $\rest{\cone}{(\para{\widehat{\ptone}^{x}}{\widehat{\pttwo}^{z}})}$ \\
    $\Pcutb{\ptone}{\cone}{\pttwo}$ & $\rightsquigarrow$ &   $\rest{\cone}(\para{\binc{\cone}{\ctwo}{\widehat{\ptone}^{y}}}{\widehat{\pttwo}^{z}})$\\
    $\Pcutd{\ptone}{\cone}{\pttwo}$ & $\rightsquigarrow$ &  $\rest{\cone}(\para{\binc{\cone}{\ctwo}{\widehat{\ptone}^{y}}}{\widehat{\pttwo}^{z}})$ \\
    $\Pbemb{\cone}{\ctwo}{\pttwo}$  & $\rightsquigarrow$ &  $\rest{\ctwo}{\outc{\cone}{\ctwo}{\widehat{\pttwo}^{z}}}$\\ 
    $\Pbemd{\cone}{\ctwo}{\pttwo}$  & $\rightsquigarrow$ &  $\rest{\ctwo}{\outc{\cone}{\ctwo}{\widehat{\pttwo}^{z}}}$ \\ 
    $\PRbang{\ptone}{\cone_1,\ldots,\cone_n}$  & $\rightsquigarrow$ &   $\binc{\cone}{\ctwo}{\widehat{\ptone}^{y}}$\\
    $\PLbangb{\cone}{\ptone}$  & $\rightsquigarrow$ & $\widehat{\ptone}^{z}$  \\
    $\PLbangd{\cone}{\ptone}$  & $\rightsquigarrow$ & $\widehat{\ptone}^{z}$   \\
    $\PLplus{\cone}{\ctwo}{\ptone}{\cthree}{\pttwo}$  & $\rightsquigarrow$ &  $\case{\ctwo}{\widehat{\ptone}^{x}}{\widehat{\pttwo}^{z}}$ \\
    $\PRplusone{\ptone}$ & $\rightsquigarrow$ &    $\inl{\cone}{\widehat{\ptone}^{x}}$\\
    $\PRplustwo{\ptone}$ & $\rightsquigarrow$ &    $\inr{\ctwo}{\widehat{\ptone}^{y}}$ \\
    $\PLwithone{\cone}{\ctwo}{\pttwo}$  & $\rightsquigarrow$  &  $\inl{\cone}{\widehat{\ptone}^{z}}$ \\
    $\PLwithtwo{\cone}{\ctwo}{\ptone}$  &  $\rightsquigarrow$ &  $\inr{\ctwo}{\widehat{\ptone}^{z}}$ \\
    $\PRwith{\ptone}{\pttwo}$  &  $\rightsquigarrow$  &   $\case{\cthree}{\widehat{\ptone}^{z}}{\widehat{\pttwo}^{z}}$\\
  \end{tabular}
\end{center}
\caption{Extraction of processes from proof terms. }\label{fig:typextr}
\end{figure}
}{}

Subject reduction will be proved by showing that if $\procone$ is typable by a type derivation
$\ptone$ and $\procone\reds\proctwo$, then a type derivation $\pttwo$ for $\proctwo$ exists. Actually,
$\pttwo$ can be obtained by manipulating $\ptone$ using techniques derived from cut-elimination.
Noticeably, not every cut-elimination rule is necessary to prove subject reduction. In other words,
we are in presence of a weak correspondence between proof terms and processes, and remain
far from a genuine Curry-Howard correspondence. 

Those manipulations of proof-terms which are necessary to prove subject reduction can be classified
as follows:
\begin{varitemize}
\item
  First of all, a binary relation $\cred$ on proof terms called \emph{computational reduction} can
  be defined. At the logical level, this corresponds to proper cut-elimination steps, i.e. those
  cut-elimination steps in which two rules introducing the same connective interact. At the
  process level, computational reduction correspond to internal reduction. $\cred$ is not
  symmetric. \condinc{Computational reduction rules are given in Figure \ref{fig:comred}.}{}
\item
  A binary relation $\cpred$ on proof terms called \emph{shift reduction}, distinct from $\cred$
  must be introduced. At the process level, it corresponds to structural congruence.
  As $\cred$, $\cpred$ is not a symmetric relation. \condinc{Shift reduction rules are given in Figure \ref{fig:shiftred}.}{}
\item
  Finally, an equivalence relation $\cequ$ on proof terms called \emph{proof equivalence} is necessary.
  At the logical level, this corresponds to the so-called commuting conversions, while at the process
  level, the induced processes are either structurally congruent or strongly bisimilar. \condinc{Equivalence rules are given in Figure \ref{fig:equred}.}{}
\end{varitemize}
\condinc{
\begin{figure}
  \begin{center}
    \begin{tabular}[t]{crcl}
      $(\cut/\Rten/\Lten):$ & $\Pcut{(\PRten{\ptone}{\pttwo})}{\cone}{\PLten{\cone}{\ctwo}{\cone}{\ptthree}}$  & $\cred$ &
      $\Pcut{\ptone}{\ctwo}{\Pcut{\pttwo}{\cone}{\ptthree}}$\\

      $(\cut/\Llin/\Rlin):$ &
      $\Pcut{\PRlin{\ctwo}{\ptone}}{\cone}{\PLlin{\cone}{\pttwo}{\cone}{\ptthree}}$ &  $\cred$ &
      $\Pcut{\Pcut{\pttwo}{\ctwo}{\ptone}}{\cone}{\ptthree}$ \\

      $(\cut/\Rwith/\Lwithone):$ &
      $\Pcut{\PRwith{\ptone}{\pttwo}}{\cone}{\PLwithone{\cone}{\ctwo}{\ptthree}}$ &
      $\cred$ &  
      $\Pcut{\ptone}{\cone}{\ptthree}$\\

      $(\cut/\Rwith/\Lwithtwo):$ &
      $\Pcut{\PRwith{\ptone}{\pttwo}}{\cone}{\PLwithtwo{\cone}{\ctwo}{\ptthree}}$ & 
      $\cred$ &  
      $\Pcut{\pttwo}{\cone}{\ptthree}$\\

      $(\cut/\Rplusone/\Lplus):$ &
      $\Pcut{\PRplusone{\ptone}}{\cone}{\PLplus{\cone}{\ctwo}{\pttwo}{\cthree}{\ptthree}}$ & $\cred$ & 
      $\Pcut{\ptone}{\cone}{\pttwo}$\\
     
      $(\cut/\Rplustwo/\Lplus):$ &
      $\Pcut{\PRplustwo{\ptone}}{\cone}{\PLplus{\cone}{\ctwo}{\pttwo}{\cthree}{\ptthree}}$ &
      $\cred$ &
      $\Pcut{\ptone}{\cone}{\ptthree}$\\

      $(\cutb/-/\bemb):$ & 
      $\Pcutb{\ptone}{\cone}{\Pbemb{\cone}{\ctwo}{\pttwo}}$ & 
      $\cred$ &
      $\Pcut{\lift{\ptone}}{\ctwo}{\Pcutd{\ptone}{\cone}{\lift{\pttwo}}}$\\

      $(\cutd/-/\bemd):$ &
      $\Pcutd{\ptone}{\cone}{\Pbemd{\cone}{\ctwo}{\pttwo}}$ &
      $\cred$ & 
      $\Pcut{\lift{\ptone}}{\ctwo}{\Pcutd{\ptone}{\cone}{\pttwo}}$\\
    \end{tabular}
  \end{center}
  \caption{Computational reduction rules}\label{fig:comred}
\end{figure}
\begin{figure}
  \begin{center}
    \begin{tabular}[t]{llll}
      $(\cut/\Rbang/\Lbangb):$ &
      $\Pcut{\PRbang{\ptone}{\cone_1,\ldots,\cone_n}}{\cone}{\PLbangb{\cone}{\pttwo}}$ & 
      $\cpred$ &   
      $\PLbangb{\cone_1}{\PLbangb{\cone_2}{\ldots\PLbangb{\cone_n}{\Pcutd{\ptone}{\ctwo}{\pttwo}}\ldots}}$\\

      $(\cut/\Rbang/\Lbangd):$ &
      $\Pcut{\PRbang{\ptone}{\cone_1,\ldots,\cone_n}}{\cone}{\PLbangd{\cone}{\pttwo}}$ & 
      $\cpred$ &   
      $\PLbangb{\cone_1}{\PLbangb{\cone_2}{\ldots\PLbangb{\cone_n}{\Pcutd{\ptone}{\ctwo}{\pttwo}}\ldots}}$\\ 
    \end{tabular}
  \end{center}
  \caption{Shift reduction rules}\label{fig:shiftred}
\end{figure}
\begin{figure}
  \begin{center}
    \textbf{Structural Conversions}
  \end{center}
  \begin{center}
    \begin{tabular}[t]{crcl}
      $(\cut/-/\cut_1):$ &
      $\Pcut{\ptone}{\cone}{\Pcut{\pttwo_\cone}{\ctwo}{\ptthree_\ctwo}}$ &  $\cequ$ &
      $\Pcut{\Pcut{\ptone}{\cone}{\pttwo_\cone}}{\ctwo}{\ptthree_{\ctwo}}$\\
      $(\cut/-/\cut_2):$ &
      $\Pcut{\ptone}{\cone}{\Pcut{\pttwo}{\ctwo}{\ptthree_{\cone\ctwo}}}$ &  $\cequ$ &
      $\Pcut{\pttwo}{\cone}{\Pcut{\ptone}{\ctwo}{\ptthree_{\cone\ctwo}}}$\\
      $(\cut/-/\cutb):$ &
      $\Pcut{\ptone}{\cone}{\Pcutb{\pttwo}{\ctwo}{\ptthree_{\cone\ctwo}}}$ &  $\cequ$ &
      $\Pcutb{\pttwo}{\ctwo}{\Pcut{\ptone}{\cone}{\ptthree_{\cone\ctwo}}}$\\
      $(\cut/\cutb/-):$&
      $\Pcut
      {\Pcutb{\ptone}{\ctwo}{\pttwo_{\ctwo}}}
      {\cone}{\ptthree_{\cone}}$ &
      $\cequ$ & 
      $\Pcutb
      {\ptone}{\ctwo}{\Pcut{{\pttwo_{\ctwo}}}
        {\cone}{\ptthree_{\cone}}}$\\
      $(\cut/-/\cutd):$ &
      $\Pcut{\ptone}{\cone}{\Pcutd{\pttwo}{\ctwo}{\ptthree_{\cone\ctwo}}}$ & $\cequ$ & 
      $\Pcutd{\pttwo}{\ctwo}{\Pcut{\ptone}{\cone}{\ptthree_{\cone\ctwo}}}$\\
      $(\cut/\cutd/-):$&
      $\Pcut
      {\Pcutd{\ptone}{\ctwo}{\pttwo_{\ctwo}}}
      {\cone}{\ptthree_{\cone}}$&
      $\cequ$ & 
      $\Pcutd
      {\ptone}{\ctwo}{\Pcut{{\pttwo_{\ctwo}}}
        {\cone}{\ptthree_{\cone}}}$\\
      $(\cut/\Rone/\Lone):$&
      $\Pcut
      {\PRone}
      {\cone}{\PLone{\cone}{\ptone}}$&
      $\cequ$ &
      $\ptone$
    \end{tabular}
  \end{center}
  \begin{center}
    \textbf{Strong Bisimilarities}
  \end{center}
  \begin{center}
    \begin{tabular}[t]{crcl}
      $(\cutd/-/\cut):$ &
      $\Pcutd
      {\ptone}{\cone}{\Pcut{{\pttwo_{\cone}}}
        {\ctwo}{\ptthree_{\cone \ctwo}}}$
      & $\cequ$ & 
      $\Pcut
      {\Pcutd{\ptone}{\cone}{\pttwo_{\cone}}}
      {\ctwo}
      {\Pcutd{\ptone}{\cone}{\ptthree_{\cone \ctwo}}}$\\
      $(\cutd/- /\cutd):$&
      $\Pcutd
      {\ptone}{\cone}{\Pcutd{{\pttwo_{\cone}}}
        {\ctwo}{\ptthree_{\cone \ctwo}}}$
      & $\cequ$ & $\Pcutd
      {\ptone}
      {\cone}
      {\Pcutd{\pttwo_{\cone}}{\ctwo}
        {\Pcutd{\ptone}{\cone}{\ptthree_{\cone \ctwo}}}}$\\
      $(\cutd/- /\cutb):$ &
      $\Pcutd
      {\ptone}{\cone}{\Pcutb{{\pttwo_{\cone}}}
        {\ctwo}{\ptthree_{\cone \ctwo}}}$
      & $\cequ$ & $\Pcutb
      {\pttwo_{\cone}}
      {\ctwo}
      {\Pcutd{\ptone}{\cone}{\ptthree_{\cone \ctwo}}}
      $\\
      $(\cutb/-/\cut_1):$ &
      $\Pcutb{\ptone}{\cone}{\Pcut{\pttwo_\cone}{\ctwo}{\ptthree_\ctwo}}$ &
      $\cequ$ &
      $\Pcut{\Pcutb{\ptone}{\cone}{\pttwo_\cone}}{\ctwo}{\ptthree_{\ctwo}}$ \\
      $(\cutb/-/\cut_2):$&
      $\Pcutb{\ptone}{\cone}{\Pcut{\pttwo}{\ctwo}{\ptthree_{\cone\ctwo}}}$ &$\cequ$ & 
      $\Pcut{\pttwo}{\ctwo}{\Pcutb{\ptone}{\cone}{\ptthree_{\cone\ctwo}}}$ \\
      $(\cutb/-/\cutb)_1:$&
      $\Pcutb{\ptone}{\cone}{\Pcutb{\pttwo_\cone}{\ctwo}{\ptthree_\ctwo}}$ & $\cequ$&
      $\Pcutb{\Pcutb{\ptone}{\cone}{\pttwo_\cone}}{\ctwo}{\ptthree_{\ctwo}}$\\
      $(\cutb/-/\cutb)_2:$&
      $\Pcutb{\ptone}{\cone}{\Pcutb{\pttwo}{\ctwo}{\ptthree_{\cone\ctwo}}}$ &$\cequ$&
      $\Pcutb{\pttwo}{\cone}{\Pcutb{\ptone}{\ctwo}{\ptthree_{\cone\ctwo}}}$\\
      $(\cutb/- /\cutd):$&
      $\Pcutb
      {\ptone}{\cone}{\Pcutd{{\pttwo_{\cone}}}
        {\ctwo}{\ptthree_{\cone \ctwo}}}$
      & $\cequ$ &
      $\Pcutd
      {\pttwo_{\cone}}
      {\ctwo}
      {\Pcutb{\ptone}{\cone}{\ptthree_{\cone \ctwo}}}$\\
      $(\cutd/- /\cutd)_0:$&
      $\Pcutd
      {\ptone}{\cone}{\Pcutd{{\pttwo_{\cone}}}
        {\ctwo}{\ptthree_{\cone \ctwo}}}$
      & $\cequ$ &
      $\Pcutd
      {\pttwo_{\cone}}
      {\ctwo}
      {\Pcutd{\ptone}{\cone}{\ptthree_{\cone \ctwo}}}$ (if $\ctwo \notin FV(\widehat{\ptthree}?)$)\\
      $(\cutd/-/-_{0}):$ &
      $\Pcutd{\ptone}{\cone}{\pttwo}$ & $\cequ$ &
      $\pttwo$  (if $x \notin FN(\widehat{\pttwo})$)\\
    \end{tabular}  
  \end{center}
  \begin{center}
    \textbf{Commuting Conversions}
  \end{center}
  \begin{center}
    \begin{tabular}[t]{crcl}
      $(\cut/-/\Lone):$&
      $\Pcut{\ptone}{\cone}{\PLone{\ctwo}{\pttwo_{\cone}}}$ & $\cequ$  & $\PLone{\ctwo}{\Pcut{\ptone}{\cone}{\pttwo_{\cone}}} $\\
      $(\cut/-/\Lbangb):$&
      $\Pcut{\ptone}{\cone}{\PLbangb{\ctwo}{\pttwo_{\cone \cthree}}}$ & $\cequ$  & $\PLbangb{\ctwo}{\Pcut{\ptone}{\cone}{\pttwo_{\cone \cthree}}} $\\
      $(\cut/-/\Lbangd):$&
      $\Pcut{\ptone}{\cone}{\PLbangd{\ctwo}{\pttwo_{\cone \cthree}}}$ & $\cequ$  & $\PLbangd{\ctwo}{\Pcut{\ptone}{\cone}{\pttwo_{\cone \cthree}}} $\\
      $(\cut/\Lone/-):$&
      $\Pcut{\PLone{\ctwo}{\ptone}}{\cone}{\pttwo_{\cone}}$ & $\cequ$  & $\PLone{\ctwo}{\Pcut{\ptone}{\cone}{\pttwo_{\cone}}} $\\
      $(\cut/\Lbangb/-):$&
      $\Pcut{\PLbangb{\ctwo}{\ptone_{\cthree}}}{\cone}{\pttwo_{\cone}}$ & $\cequ$  & $\PLbangb{\ctwo}{\Pcut{\ptone_{\cthree}}{\cone}{\pttwo_{\cone \cthree}}} $\\
      $(\cut/\Lbangd/-):$&
      $\Pcut{\PLbangd{\ctwo}{\ptone_{\cthree}}}{\cone}{\pttwo_{\cone}}$ & $\cequ$  & $\PLbangd{\ctwo}{\Pcut{\ptone_{\cthree}}{\cone}{\pttwo_{\cone \cthree}}} $\\
      $(\cutb/-/\Lone):$&
      $\Pcutb{\ptone}{\cone}{\PLone{\ctwo}{\pttwo_{\cone}}}$ & $\cequ$  & $\PLone{\ctwo}{\Pcutb{\ptone}{\cone}{\pttwo_{\cone}}} $\\
      $(\cutb/-/\Lbangb):$&
      $\Pcutb{\ptone}{\cone}{\PLbangb{\ctwo}{\pttwo_{\cone \cthree}}}$ & $\cequ$  & $\PLbangb{\ctwo}{\Pcutb{\ptone}{\cone}{\pttwo_{\cone \cthree}}} $\\
      $(\cutb/-/\Lbangd):$&
      $\Pcutb{\ptone}{\cone}{\PLbangd{\ctwo}{\pttwo_{\cone \cthree}}}$ & $\cequ$  & $\PLbangd{\ctwo}{\Pcutb{\ptone}{\cone}{\pttwo_{\cone \cthree}}} $\\
      $(\cutd/-/\Lone):$&
      $\Pcutd{\ptone}{\cone}{\PLone{\ctwo}{\pttwo_{\cone}}}$ & $\cequ$  & $\PLone{\ctwo}{\Pcutd{\ptone}{\cone}{\pttwo_{\cone}}} $\\
      $(\cutd/-/\Lbangb):$&
      $\Pcutd{\ptone}{\cone}{\PLbangb{\ctwo}{\pttwo_{\cone \cthree}}}$ & $\cequ$  & $\PLbangb{\ctwo}{\Pcutd{\ptone}{\cone}{\pttwo_{\cone \cthree}}} $\\
      $(\cutd/-/\Lbangd):$&
      $\Pcutd{\ptone}{\cone}{\PLbangd{\ctwo}{\pttwo_{\cone \cthree}}}$ & $\cequ$  & $\PLbangd{\ctwo}{\Pcutd{\ptone}{\cone}{\pttwo_{\cone \cthree}}} $\\
    \end{tabular}
  \end{center}\caption{Equivalence rules}\label{fig:equred}
\end{figure}
}{}
The reflexive and transitive closure of $\cpred\cup\cequ$ is denoted with $\cpredequ$,
i.e. $\cpredequ=(\cpred\cup\cequ)^*$. \condinc{To help the reader understand the rules defining $\cred$, $\cpred$ and $\cequ$, let us give some relevant examples:}
{There is not enough space here to give the rules defining $\cred$, $\cpred$ and $\cequ$. Let us give only some relevant examples:}
\begin{varitemize}
\item
  Let us consider the proof term $\ptone = \Pcut{(\PRten{\ptthree}{\ptfour})}{\cone}{\PLten{\cone}{\ctwo}{\cone}{\ptfive}}$ 
  which corresponds to the $\tens$-case of cut elimination. By a computational reduction rule, 
  $\ptone \cred \pttwo =\Pcut{\ptthree}{\ctwo}{\Pcut{\ptfour}{\cone}{\ptfive}}$. 
  From the process side, 
  $\widehat{\ptone} = \rest{\cone}{((\para{\rest{\ctwo}{\outc{\cone}{\ctwo}{(\para{\widehat{\ptthree}}
          {\widehat{\ptfour}}))}}}{\inc{\cone}{\ctwo}{\widehat{\ptfive}}})}$ and  
  $\widehat{\pttwo}= \rest{\cone}\rest{\ctwo}{(\para{(\para{\widehat{\ptthree}}{\widehat{\ptfour}})}}{\widehat{\ptfive})}$, 
  where $\widehat{\pttwo}$  is the process obtained from $\widehat{\ptone}$ by internal passing the 
  channel $\ctwo$ through the channel $\cone$.
\item
  Let $\ptone=\Pcut{\PRbang{\ptthree}{\cone_1,\ldots,\cone_n}}{\cone}{\PLbangb{\cone}{\ptfour}}$ be 
  the proof obtained by composing a proof $\ptthree$ (whose last rule is $\Rbang$) with a proof $\ptfour$ 
  (whose last rule is $\Lbangb$) through a $\cut$ rule. A shift reduction rule tells us that  
  $\ptone \cpred \pttwo=\PLbangb{\cone_1}{\PLbangb{\cone_2}{\ldots\PLbangb{\cone_n}{\Pcutb{\ptthree}
        {\ctwo}{\ptfour}}\ldots}}$, which corresponds to the opening of a box in
  $\sll$. The shift reduction does not have a corresponding reduction step at process level, since 
  $\widehat{\ptone} \cequ \widehat{\pttwo}$; nevertheless, it is defined as an asymmetric relation, 
  for technical reasons connected to the proof of bounded interaction.
\item
  Let $\ptone=\Pcutd{\ptthree}{\cone}{\Pcut{{\ptfour}}{\ctwo}{\ptfive}}$. A defining rule 
  for proof equivalence $\cequ$, states that in $\ptone$ the $\cutd$ rule can be permuted 
  over the $\cut$ rule, by duplicating $\ptthree$; namely 
  $\ptone \cequ \pttwo=\Pcut{\Pcutd{\ptthree}{\cone}{\ptfour}}{\ctwo}{\Pcutd{\ptthree}{\cone}{\ptfive}}$. 
  This is possible because the channel $\cone$ belongs to the multiplexor contexts of both $\ptfour,\ptfive$, 
  such  contexts being treated additively. At the process level, 
  $\widehat{\ptone}=\rest{\cone}{(\para{\para{(\binc{\cone}{\ctwo}{\widehat{\ptthree}})}
      {\rest{\ctwo}(\widehat{\ptfour}}}}{\widehat{\ptfive}))}$ , while 
  $\widehat{\pttwo}=\rest{\ctwo}{((\para{\rest{\cone}{\para{(\binc{\cone}{\ctwo}{\widehat{\ptthree}})}
        {\widehat{\ptfour}}}))}{(\rest{\cone}{\para{(\binc{\cone}{\ctwo}{\widehat{\ptthree}})}{\widehat{\ptfive})))}}}}$, 
  $\widehat{\ptone}$ and $\widehat{\pttwo}$ being strongly bisimilar.
\end{varitemize}
\condinc{The following propositions state the correspondences between the proof terms manipulation rules described above and relations over processes: we omit the 
proofs, leaving to the reader the verification of each case. 
\begin{proposition}
 Let $\tyg{\conone}{\contwo}{\conthree}{\ptone}{\tcone}$ and $\tyg{\conone'}{\contwo'}{\conthree'}{\pttwo}{\tcone'}$. If $\ptone \cred \pttwo$, then 
$\widehat{\ptone} \reds \widehat{\pttwo}$.
\end{proposition}
\begin{proposition}
 Let $\tyg{\conone}{\contwo}{\conthree}{\ptone}{\tcone}$ and $\tyg{\conone'}{\contwo'}{\conthree'}{\pttwo}{\tcone'}$. If $\ptone \cpred \pttwo$, then 
$\widehat{\ptone}$ is equivalent to $\widehat{\pttwo}$ modulo structural congruence.
\end{proposition}
\begin{proposition}
 Let $\tyg{\conone}{\contwo}{\conthree}{\ptone}{\tcone}$ and $\tyg{\conone'}{\contwo'}{\conthree'}{\pttwo}{\tcone'}$. If $\ptone \cequ \pttwo$, then 
$\widehat{\ptone}$ is equivalent to $\widehat{\pttwo}$ modulo structural congruence or strong bisimilarity.
\end{proposition}
}{}
\noindent
Before proceeding to Subject Reduction, we give the following two lemmas, concerning structural properties of the type system:
\begin{lemma}[Weakening lemma]\label{weaklemma}
If $\tyg{\conone}{\contwo}{\conthree}{\ptone}{\tcone}$ and
whenever $\contwo\subseteq\confour$, it holds
that $\tyg{\conone}{\confour}{\conthree}{\ptone}{\tcone}$.
\end{lemma}
\begin{proof}
 By a simple induction on the structure of $\ptone$.
\end{proof}
\begin{lemma}[Lifting lemma]\label{liftlemma}
If $\tyg{\conone}{\contwo}{\conthree}{\ptone}{\tcone}$
then there exists an  $\pttwo$ such that
$\tyg{\emcon}{\conone,\contwo}{\conthree}{\pttwo}{\tcone}$
where $\widehat{\pttwo}=\widehat{\ptone}$.
We denote $\pttwo$ by $\lift{\ptone}$.
\end{lemma}  \
\begin{proof}
Again, a simple induction on the structure of the proof term $\ptone$.
\end{proof}
Finally:
\begin{theorem}[Subject Reduction]\label{theo:subjred}
Let $\tyg{\conone}{\contwo}{\conthree}{\ptone}{\tcone}$. 
Suppose that $\widehat{\ptone}=\procone\reds\proctwo$. Then there
is $\pttwo$ such that $\widehat{\pttwo}=\proctwo$, 
$\ptone\cpredequ\cred\cpredequ\pttwo$ and
$\tyg{\confour}{\confive}{\conthree}{\pttwo}{\tcone}$, where 
$\conone,\contwo=\confour,\confive$.
\end{theorem}

\condinc{
In order to prove Theorem \ref{theo:subjred} we need the followings auxiliary results:
}
{Let us 
give a sketch of the proof of Theorem \ref{theo:subjred}. We reason by induction on 
the structure of $\ptone$. Since $\widehat{\ptone}=\procone\reds\proctwo$ the only possible last rules of
$\ptone$ can be: $\Lone, \Lbangb, \Lbangd,$, a linear cut ($\cut$) or an exponential cut ($\cutb$ or $\cutd$).
In all the other cases, the underlying process can only perform a visible action, as can be easily
verified by inspecting the rules from Figure~\ref{fig:pidslltr}. With this observation in mind, let
us inspect the operational semantics derivation proving that $\procone\reds\proctwo$. At some point
we will find two subprocesses of $\procone$, call them $\procthree$ and $\procfour$, which
communicate, causing an internal reduction. We here claim that this can only happen 
in presence of a cut, and only the communication between $\procthree$ and $\procfour$ must
occur along the channel involved in the cut. Now, it's only a matter of showing that the
just described situation can be ``resolved'' preserving types. And this can be done by 
way of several lemmas, like the following:}

\condinc{
\begin{lemma}\label{lem:visible}
Let $\tyg{\conone}{\contwo}{\conthree}{\ptone \rightsquigarrow \procone}{\cone : \tcone}$.
\begin{varenumerate}
 \item If $\procone \xrightarrow{\alpha} \proctwo$ and $\tcone = \unit$ then $s(\alpha) \neq \cone$.
 \item If $\procone \xrightarrow{\alpha} \proctwo$ and $\ctwo : \unit \in \conthree$ then $s(\alpha) \neq \ctwo$.
 \item If $\procone \xrightarrow{\alpha} \proctwo$ and $s(\alpha) = \cone$ and $\tcone = \typone \otimes \typtwo$ then $\alpha = \overline{\rest{\ctwo}{\outsc{\cone}{\ctwo}}}$.
  \item If $\procone \xrightarrow{\alpha} \proctwo$ and $s(\alpha) = \ctwo$ and $\ctwo : \typone \otimes \typtwo \in \conthree$ then $\alpha = \insc{\ctwo}{\cthree}$.
  \item If $\procone \xrightarrow{\alpha} \proctwo$ and $s(\alpha) = \cone$ and $\tcone = \typone \lin \typtwo$ then $\alpha = \insc{\cone}{\ctwo}$.
  \item If $\procone \xrightarrow{\alpha} \proctwo$ and $s(\alpha) = \ctwo$ and $\ctwo : \typone \lin \typtwo \in \conthree$ then $\alpha = \overline{\rest{\cthree}{\outsc{\ctwo}{\cthree}}}$.
  \item If $\procone \xrightarrow{\alpha} \proctwo$ and $s(\alpha) = \cone$ and $\tcone = \typone \with \typtwo$ then $\alpha = \inl{\cone}{}$ or $\alpha = \inr{\cone}{}$.
  \item If $\procone \xrightarrow{\alpha} \proctwo$ and $s(\alpha) = \ctwo$ and $\ctwo : \typone \with \typtwo \in \conthree$ then $\alpha = \overline{\inl{\ctwo}{}}$ or $\alpha = \overline{\inr{\ctwo}{}}$.
  \item If $\procone \xrightarrow{\alpha} \proctwo$ and $s(\alpha) = \cone$ and $\tcone = \typone \plus \typtwo$ then
$\alpha = \overline{\inl{\cone}{}}$ or $\alpha = \overline{\inr{\cone}{}}$.
  \item If $\procone \xrightarrow{\alpha} \proctwo$ and $s(\alpha) = \ctwo$ and $\ctwo : \typone \plus \typtwo \in \conthree$ then  $\alpha = \inl{\ctwo}{}$ or $\alpha = \inr{\ctwo}{}$
  \item If $\procone \xrightarrow{\alpha} \proctwo$ and $s(\alpha) = \cone$ and $\tcone = \bang \typone$ then $\alpha = \insc{\cone}{\ctwo}$.
  \item If $\procone \xrightarrow{\alpha} \proctwo$ and $s(\alpha) = \ctwo$ and $\ctwo : \bang \typone$ or $\ctwo \in \conone$ or $\ctwo \in \contwo$ or $\ctwo \in \confour$ then $\alpha = \overline{\rest{\cthree}{\outsc{\ctwo}{\cthree}}}$.
\end{varenumerate}
\end{lemma}
\begin{proof}
Trivial from definitions.
\end{proof}}{}

\begin{lemma}\label{lemma:tens}
 Assume that:
\begin{varenumerate}
  \item   
    $\tyg{\conone_1}{\contwo}{\conthree_1}{\ptone}{\cone : \typone\tens \typtwo}$ 
    with $\widehat{\ptone}=\procone\xrightarrow{\overline{\rest{\ctwo}{\outsc{\cone}{\ctwo}}}} \proctwo$;      
  \item 
    $\tyg{\conone_2}{\contwo}{\conthree_2,\cone : \typone \tens \typtwo}{\pttwo}{\cthree : \typthree}$ 
    with $\widehat{\pttwo}=\procthree \xrightarrow{\insc{\cone}{\ctwo}} \procfour$.
\end{varenumerate}
Then:
\begin{varenumerate}
 \item 
   $\Pcut{\ptone}{\cone}{\pttwo}\cpredequ\cred\cpredequ\ptthree$ for some $\ptthree$;
 \item 
   $\tyg{\conone_1, \conone_2}{\contwo}{\conthree_1, \conthree_2}{\ptthree}{\cthree : \typthree}$,
   where $\widehat{\ptthree}\cequ\rest{\cone}{(\para{\proctwo}{\procfour})}$. 
\end{varenumerate}
\end{lemma}
\condinc{
\begin{proof}
By simultaneous induction on $\ptone_1, \ptone_2$. The property stated in the lemma holds also for the system $\pidill$ (see ~\cite{Caires10}); 
since the proof technique is essentially the same modulo some minor details, we omit the proof.
\end{proof}
\begin{lemma}\label{lemma:lin}
 Assume 
\begin{varenumerate}
\item $\tyg{\conone_1}{\contwo}{\conthree_1}{\ptone_1 \rightsquigarrow \procone_1}{\cone : \typone \lin \typtwo}$ with $\procone_1 \xrightarrow{\insc{\cone}{\ctwo}} \procone'_1$
\item $\tyg{\conone_2}{\contwo}{\conthree_2,\cone : \typone \lin \typtwo}{\ptone_2 \rightsquigarrow \proctwo_1}{\cthree : \typthree}$ with $\proctwo_1 \xrightarrow{\overline{\rest{\ctwo}{\outsc{\cone}{\ctwo}}}} \proctwo'_1$
\end{varenumerate}
Then
\begin{varenumerate}
 \item $\Pcut{\ptone_1}{\cone}{\ptone_2} \cpredequ\cred\cpredequ \ptone $ for some $\ptone$;
 \item $\tyg{\conone_1, \conone_2}{\contwo}{\conthree_1, \conthree_2}{\ptone \rightsquigarrow \proctwo_2}{\cthree : \typthree}$ for some $\proctwo_2 \cequ \rest{\cone}{\rest{\ctwo}{\para{(\procone'_1}{\proctwo'_1)}}}$. 
\end{varenumerate}
\end{lemma}
\begin{proof}
See the proof of Lemma \ref{lemma:tens}.
\end{proof}

\begin{lemma}\label{lemma:bang}
 Assume 
\begin{varenumerate}
\item $\tyg{\conone_1}{\contwo}{\conthree_1}{\ptone_1 \rightsquigarrow \procone_1}{\cone : \bang \typone }$ with $\procone_1 \xrightarrow{\insc{\cone}{\ctwo}} \procone'_1$
\item $\tyg{\conone_2}{\contwo}{\conthree_2,\cone : \bang \typone}{\ptone_2 \rightsquigarrow \proctwo_1}{\cthree : \typthree}$ with $\proctwo_1 \xrightarrow{\overline{\rest{\ctwo}{\outsc{\cone}{\ctwo}}}} \proctwo'_1$
\end{varenumerate}
Then
\begin{varenumerate}
 \item $\Pcut{\ptone_1}{\cone}{\ptone_2} \cpredequ\cred\cpredequ \ptone $ for some $\ptone$;
 \item $\tyg{\conone_1, \conone_2}{\contwo}{\conthree_1, \conthree_2}{\ptone \rightsquigarrow \proctwo_2}{\cthree : \typthree}$ for some $\proctwo_2 \cequ \rest{\cone}{\rest{\ctwo}{\para{(\procone'_1}{\proctwo'_1)}}}$. 
\end{varenumerate}

\end{lemma}

\begin{proof}
See the proof of Lemma \ref{lemma:tens}.
\end{proof}

\begin{lemma}\label{lemma:with}
 Assume 
\begin{varenumerate}
\item $\tyg{\conone_1}{\contwo}{\conthree_1}{\ptone_1 \rightsquigarrow \procone_1}{\cone : \typone \with \typtwo}$ with $\procone_1 \xrightarrow{\inl{\cone}{}} \procone'_1$
\item $\tyg{\conone_2}{\contwo}{\conthree_2,\cone : \typone \with \typtwo}{\ptone_2 \rightsquigarrow \proctwo_1}{\cthree : \typthree}$ with $\proctwo_1 \xrightarrow{\overline{\inl{\cone}{}}} \proctwo'_1$
\end{varenumerate}
Then
\begin{varenumerate}
 \item $\Pcut{\ptone_1}{\cone}{\ptone_2} \cpredequ\cred\cpredequ \ptone $ for some $\ptone$;
 \item $\tyg{\conone_1, \conone_2}{\contwo}{\conthree_1, \conthree_2}{\ptone \rightsquigarrow \proctwo_2}{\cthree : \typthree}$ for some $\proctwo_2 \cequ \rest{\cone}{\para{(\procone'_1}{\proctwo'_1)}}$. 
\end{varenumerate}

\end{lemma}

\begin{proof}
See the proof of Lemma \ref{lemma:tens}.
\end{proof}

\begin{lemma}\label{lemma:plus}
 Assume 
\begin{varenumerate}
\item $\tyg{\conone_1}{\contwo}{\conthree_1}{\ptone_1 \rightsquigarrow \procone_1}{\cone : \typone \plus \typtwo}$ with $\procone_1 \xrightarrow{\overline{\inl{\cone}{}}} \procone'_1$
\item $\tyg{\conone_2}{\contwo}{\conthree_2,\cone : \typone \plus \typtwo}{\ptone_2 \rightsquigarrow \proctwo_1}{\cthree : \typthree}$ with $\proctwo_1 \xrightarrow{\inl{\cone}{}} \proctwo'_1$
\end{varenumerate}
Then
\begin{varenumerate}
 \item $\Pcut{\ptone_1}{\cone}{\ptone_2} \cpredequ\cred\cpredequ \ptone $ for some $\ptone$;
 \item $\tyg{\conone_1, \conone_2}{\contwo}{\conthree_1, \conthree_2}{\ptone \rightsquigarrow \proctwo_2}{\cthree : \typthree}$ for some $\proctwo_2 \cequ \rest{\cone}{\para{(\procone'_1}{\proctwo'_1)}}$. 
\end{varenumerate}

\end{lemma}

\begin{proof}
See the proof of Lemma \ref{lemma:tens}.
\end{proof}

\begin{lemma}\label{lemma:cutbang}
 Assume 
\begin{varenumerate}
\item $\tyg{\conone_1}{\emcon}{\emcon}{\ptone_1 \rightsquigarrow \procone_1}{\cone : \typone}$
\item $\tyg{\conone_2, \cone : \typone}{\contwo}{\conthree}{\ptone_2 \rightsquigarrow \proctwo_1}{\cthree : \typthree}$ with $\proctwo_1 \xrightarrow{\overline{\rest{\ctwo}{\outsc{\cone}{\ctwo}}}} \proctwo'_1$
\end{varenumerate}
Then
\begin{varenumerate}
 \item $\Pcutb{\ptone_1}{\cone}{\ptone_2} \cpredequ\cred\cpredequ \Pcutd{\ptone_1}{\cone}{\ptone} $ for some $\ptone$ where $\cone \notin FV(\widehat{\ptone})$;
 \item $\tyg{\conone}{\confour}{\conthree}{\ptone \rightsquigarrow \proctwo_2}{\cthree : \typthree}$ for some $\proctwo_2 \cequ \rest{\ctwo}{\para{(\procone_1}{\proctwo'_1)}}$, where $\conone, \confour =\conone_1, \conone_2, \cone : \typone, \contwo$.
\end{varenumerate}

\end{lemma}

 \begin{proof}
By  induction on $ \ptone_2$. We have different cases, depending from the last rules of $\ptone_2$. Let us just write down some relevant case:

\begin{varitemize}
\item 
  Suppose $\ptone_2 = \Pbemb{\cone}{\ctwo}{\ptone'_2}$; then  $\proctwo_1 \cequ \rest{\ctwo}{\outc{\cone}{\ctwo}{\proctwo'_1}}$ and 
  $\tyg{ \conone_2, \cone : \typone}{\contwo}{\conthree}{\ptone'_2 \rightsquigarrow \proctwo'_1}{\cthree : \typthree}$ by inversion. Now 
  $ \Pcutb{\ptone_1}{\cone}{\Pbemb{\cone}{\ctwo}{\ptone'_2}} \cred \Pcut{\ptone_{1 \Downarrow }}{\ctwo}{\Pcutd{\ptone_1}{\cone}{\ptone'_{2 \Downarrow}}}$ 
  by ($\cutb/- /\bemb$) $\cequ \Pcutd{\ptone_1}{\cone}{\Pcut{\ptone_{1 \Downarrow}}{\ctwo}{\ptone'_{2 \Downarrow}}}$ by ($\cut/- /\cutd$). We pick 
  $D = \Pcut{\ptone_{1 \Downarrow}}{\ctwo}{\ptone'_{2 \Downarrow}}$; then $\tyg{\conone}{\confour}{\conthree}{\ptone \rightsquigarrow \proctwo_2}{\cthree : \typthree}$ 
  for some $\proctwo_2 \cequ \rest{\ctwo}{\para{(\procone_1}{\proctwo'_1)}}$,  where $\conone, \confour = \conone_1, \conone_2, \cone : \typone, \contwo$.
\item 
  Suppose $\ptone_2 = \Pcutd{\ptone'_1}{\ctwo}{\ptone'_2}$; then $\tyg{\contwo}{\emcon}{\emcon}{\ptone'_1 \rightsquigarrow \procthree_1}{\cfour : \typthree}$ 
  and $\tyg{ \conone_2, \cone : \typone}{\contwo}{\conthree}{\ptone'_2 \rightsquigarrow \procthree'_2}{\cthree : \typtwo}$ with 
  $\proctwo_1 \xrightarrow{\overline{\rest{\ctwo}{\outsc{\cone}{\ctwo}}}} \para{\procthree_1}{\procthree'_2}$,  by inversion.
  Now by induction hypothesis, $\Pcutb{\ptone_1}{\cone}{\ptone_2'} \cpredequ\cred\cpredequ \Pcutd{\ptone_1}{\cone}{\ptone^{*}} $ for some 
  $\ptone^{*}$ (where $\cone \notin FV(\widehat{\ptone^{*}})$, and $\tyg{\conone}{\confour}{\conthree_2}{\ptone^{*} \rightsquigarrow \procfour}{\cthree : \typtwo}$ 
  for some $\procfour = \rest{\ctwo}{\para{(\procone_1}{\procthree'_2)}}$. $\Pcutb{\ptone_1}{\cone}{\Pcutd{\ptone'_1}{\ctwo}{\ptone'_2}} \cequ 
  \Pcutd{\ptone'_1}{\ctwo}{\Pcutb{\ptone_1}{\cone}{\ptone'_2}}$ by $(\cutb/-/\cutd)$, $\cpredequ\cred\cpredequ \Pcutd{\ptone'_1}{\ctwo}{\Pcutd{\ptone_1}{\cone}{\ptone^{*}}}$  
  by congruence, $\cequ \Pcutd{\ptone_1}{\cone}{\Pcutd{\ptone'_1}{\ctwo}{\ptone^{*}}}$ by $(\cutd/-/\cutd)_0$. Pick $D = \Pcutd{\ptone'_1}{\ctwo}{\ptone^{*}}$. 
  Then $\proctwo_2=  \rest{\ctwo}{\para{\procthree_1}{\procfour}}$ by cut, and $\tyg{\conone}{\confour}{\conthree}{\ptone \rightsquigarrow \proctwo_2}{\cthree : \typthree}$ 
  for some $\proctwo_2 \cequ \rest{\ctwo}{\para{(\procone_1}{\proctwo'_1)}}$. 
\end{varitemize}
This concludes the proof.
\end{proof}

\begin{corollary}\label{cor:cutb}
 Assume 
\begin{varenumerate}
\item $\tyg{\conone_1}{\emcon}{\emcon}{\ptone_1 \rightsquigarrow \procone_1}{\cone : \typone}$
\item $\tyg{\conone_2, \cone : \typone}{\contwo}{\conthree}{\ptone_2 \rightsquigarrow \proctwo_1}{\cthree : \typthree}$ with $\proctwo_1 \xrightarrow{\overline{\rest{\ctwo}{\outsc{\cone}{\ctwo}}}} \proctwo'_1$
\end{varenumerate}
Then
\begin{varenumerate}
 \item $\Pcutb{\ptone_1}{\cone}{\ptone_2} \cpredequ\cred\cpredequ \ptone $ for some $\ptone$;
 \item $\tyg{\conone}{\confour}{\conthree}{\ptone \rightsquigarrow \proctwo_2}{\cthree : \typthree}$ for some $\proctwo_2 \cequ \rest{\cone}{(\binc{\cone}{\ctwo}{\para{\procone_1}{\rest{\ctwo}{\para{(\procone_1}{\proctwo'_1))}}}}}$,  where $\conone, \confour =\conone_1, \conone_2, \contwo$ 
\end{varenumerate}

\end{corollary}

\begin{proof}
 Follows from Lemma \ref{lemma:cutbang}.
\end{proof}

\begin{lemma}\label{lemma:cutd}
 Assume 
\begin{varenumerate}
\item $\tyg{\contwo}{\emcon}{\emcon}{\ptone_1 \rightsquigarrow \procone_1}{\cone : \typone}$
\item $\tyg{\conone}{\contwo, \cone : \typone}{\conthree}{\ptone_2 \rightsquigarrow \proctwo_1}{\cthree : \typthree}$ with $\proctwo_1 \xrightarrow{\overline{\rest{\ctwo}{\outsc{\cone}{\ctwo}}}} \proctwo'_1$
\end{varenumerate}
Then :

\begin{varenumerate}
 \item $\Pcutd{\ptone_1}{\cone}{\ptone_2} \cpredequ\cred\cpredequ \Pcutd{\ptone_1}{\cone}{\ptone} $ for some $\ptone$;
 \item $\tyg{\confour}{\confive, \cone : \typone}{\conthree}{\ptone \rightsquigarrow \proctwo_2}{\cthree : \typthree}$ for some $\proctwo_2 \cequ \rest{\cone}{\rest{\ctwo}{\para{(\procone_1}{\proctwo'_1)}}}$, where $\confour, \confive =\conone, \contwo$. 
\end{varenumerate}

\end{lemma}

\begin{proof}

By  induction on $ \ptone_2$. We have different cases, depending from the last rules of $\ptone_2$. Let us just write down some relevant cases:
\begin{varitemize}
\item 
  $\ptone_2 = \Pcut{\ptone'_1}{\ctwo}{\ptone'_2}$. Assume $\conone = \conone_1, \conone_2$ and $\conthree = \conthree_1, \conthree_2$. 
  Now $\tyg{\conone_1}{\contwo, \cone : \typone}{\conthree_1}{\ptone'_1 \rightsquigarrow \procthree_1}{\cfour : \typtwo}$ and 
  $\tyg{ \conone_2}{\contwo, \cone : \typone}{\conthree, \cfour : \typtwo}{\ptone'_2 \rightsquigarrow \procthree_2}{\cthree : \typthree}$ 
  by inversion. We have two cases:either $\proctwo_1 \xrightarrow{\overline{\rest{\ctwo}{\outsc{\cone}{\ctwo}}}} \para{\procthree_1}{\procthree'_2}$, 
  or $\proctwo_1 \xrightarrow{\overline{\rest{\ctwo}{\outsc{\cone}{\ctwo}}}} \para{\procthree'_1}{\procthree_2}$.
  First case:
  $\Pcutd{\ptone_1}{\cone}{\ptone'_1} \cpredequ\cred\cpredequ \Pcutd{\ptone_1}{\cone}{\ptone^{*}} $ for some $\ptone^{*}$; then      
  $\tyg{\conone_1}{\contwo, \cone : \typone}{\conthree_1}{\ptone^{*} \rightsquigarrow \procfour}{\cfour : \typtwo}$ for some 
  $\procfour = \rest{\ctwo}{\para{(\procone_1}{\procthree'_1)}}$ by induction hypothesis;
  $\Pcutd{\ptone_1}{\cone}{\Pcut{\ptone'_1}{\ctwo}{\ptone'_2}} \cequ 
  \Pcut
  {\Pcutd{\ptone_1}{\cone}{\ptone'_1}}
  {\ctwo}
  {\Pcutd{\ptone_1}{\cone}{\ptone'_2}}$ by $(\cutd/-/\cut)$, $\cpredequ\cred\cpredequ 
  \Pcut
  {\Pcutd{\ptone_1}{\cone}{\ptone^{*}}}
  {\ctwo}
  {\Pcutd{\ptone_1}{\cone}{\ptone'_2}}$  by congruence 
  $\cequ \Pcutd{\ptone_1}{\cone}{\Pcut{\ptone^{*}}{\ctwo}{\ptone'_2}}$ by $(\cutd/-/\cut)$. 
  Pick $D = \Pcut{\ptone^{*}}{\ctwo}{\ptone'_2}$; then $\proctwo_2=  \rest{\ctwo}{\para{\procfour}{\procthree_2}}$ by cut.
  Then $\tyg{\conone}{\contwo, \cone : \typone}{\conthree}{\ptone \rightsquigarrow \proctwo_2}{\cthree : \typthree}$ for some $\proctwo_2 \cequ \rest{\ctwo}{\para{(\procone_1}{\proctwo'_1)}}$. 
  Second case:
  $\Pcutd{\ptone_1}{\cone}{\ptone'_2} \cpredequ\cred\cpredequ \Pcutd{\ptone_1}{\cone}{\ptone^{*}} $ for some $\ptone^{*}$;
  then $\tyg{\conone_2}{\contwo, \cone : \typone}{\conthree_2}{\ptone^{*} \rightsquigarrow \procfour}{\cfour : \typtwo}$ for some $\procfour = \rest{\ctwo}{\para{(\procone_1}{\procthree'_2)}}$ 
  by induction hypothesis;
  $\Pcutd{\ptone_1}{\cone}{\Pcut{\ptone'_1}{\ctwo}{\ptone'_2}} \cequ 
  \Pcut
  {\Pcutd{\ptone_1}{\cone}{\ptone'_1}}
  {\ctwo}
  {\Pcutd{\ptone_1}{\cone}{\ptone'_2}}$ by $(\cutd/-/\cut)$, $\cpredequ\cred\cpredequ 
  \Pcut
  {\Pcutd{\ptone_1}{\cone}{\ptone'_1}}
  {\ctwo}
  {\Pcutd{\ptone_1}{\cone}{\ptone^{*}}}$  by congruence, $\cequ \Pcutd{\ptone_1}{\cone}{\Pcutd{\ptone'_1}{\ctwo}{\ptone^{*}}}$ by $(\cutd/-/\cut).$
  Pick $D = \Pcutd{\ptone'_1}{\ctwo}{\ptone^{*}}$; then $\proctwo_2=  \rest{\ctwo}{\para{\procthree_1}{\procfour}}$ by cut.
  Then $\tyg{\conone}{\contwo, \cone : \typone}{\conthree}{\ptone \rightsquigarrow \proctwo_2}{\cthree : \typthree}$ for some $\proctwo_2 \cequ \rest{\ctwo}{\para{(\procone_1}{\proctwo'_1)}}$. 
\item 
  $\ptone_2 = \Pcutd{\ptone'_1}{\ctwo}{\ptone'_2}$. $\tyg{\contwo}{\emcon}{\emcon}{\ptone'_1 \rightsquigarrow \procthree_1}{\cfour : \typtwo}$ 
  $\tyg{ \conone}{\contwo, \cone : \typone, \cfour : \typtwo}{\conthree}{\ptone'_2 \rightsquigarrow \procthree_2}{\cthree : \typthree}$ by inversion.
  Now $\proctwo_1 \xrightarrow{\overline{\rest{\ctwo}{\outsc{\cone}{\ctwo}}}} \para{\procthree_1}{\procthree'_2}$; 
  $\Pcutd{\ptone_1}{\cone}{\ptone'_2} \cpredequ\cred\cpredequ \Pcutd{\ptone_1}{\cone}{\ptone^{*}} $ for some $\ptone^{*}$ and 
  $\tyg{\conone}{\contwo, \cone : \typone, \cfour : \typtwo}{\conthree}{\ptone^{*} \rightsquigarrow \procfour}{\cfour : \typtwo}$ for some 
  $\procfour = \rest{\ctwo}{\para{(\procone_1}{\procthree'_2)}}$ by induction hypothesis. $\Pcutd{\ptone_1}{\cone}{\Pcutd{\ptone'_1}{\ctwo}{\ptone'_2}} \cequ 
  \Pcutd
  {\ptone_1}
  {\cone}
  {\Pcutd{\ptone'_1}{\ctwo}
    {\Pcutd{\ptone_1}{\cone}{\ptone'_2}}}$ by $(\cutd/- /\cutd)$ $\cpredequ\cred\cpredequ 
  \Pcutd
  {\ptone_1}
  {\cone}
  {\Pcutd{\ptone'_1}{\ctwo}
    {\Pcutd{\ptone_1}{\cone}{\ptone^{*}}}}$  by congruence, $\cequ \Pcutd{\ptone_1}{\cone}{\Pcutd{\ptone'_1}{\ctwo}{\ptone^{*}}}$ by $(\cutd/- /\cutd).$ 
  Pick $D = \Pcutd{\ptone'_1}{\ctwo}{\ptone^{*}}$; then  $\proctwo_2=  \rest{\ctwo}{\para{\procthree_1}{\procfour}}$ by cut.
  Then $\tyg{\conone}{\contwo, \cone : \typone}{\conthree}{\ptone \rightsquigarrow \proctwo_2}{\cthree : \typthree}$ for some $\proctwo_2 \cequ \rest{\ctwo}{\para{(\procone_1}{\proctwo'_1)}}$. 
\end{varitemize}
This concludes the proof.
\end{proof}

\begin{corollary}\label{cor:cutd}
 Assume 
\begin{varenumerate}
\item $\tyg{\contwo}{\emcon}{\emcon}{\ptone_1 \rightsquigarrow \procone_1}{\cone : \typone}$
\item $\tyg{\conone}{\cone : \typone, \contwo}{\conthree}{\ptone_2 \rightsquigarrow \proctwo_1}{\cthree : \typthree}$ with $\proctwo_1 \xrightarrow{\overline{\rest{\ctwo}{\outsc{\cone}{\ctwo}}}} \proctwo'_1$
\end{varenumerate}
Then
\begin{varenumerate}
 \item $\Pcutd{\ptone_1}{\cone}{\ptone_2} \cpredequ\cred\cpredequ \ptone $ for some $\ptone$;
 \item $\tyg{\confour}{\confive}{\conthree}{\ptone \rightsquigarrow \proctwo_2}{\cthree : \typthree}$ for some $\proctwo_2 \cequ \rest{\cone}{(\binc{\cone}{\ctwo}{\para{\procone_1}{\rest{\ctwo}{\para{(\procone_1}{\proctwo'_1))}}}}}$, where $\confour, \confive =\conone, \contwo$.
\end{varenumerate}
\end{corollary}

\begin{proof}
 This follows from Lemma \ref{lemma:cutd}.
\end{proof}
}{The other lemmas can be found in \cite{DalLagoDiGiambe2011ev}. By the way, this proof technique
is very similar to the one introduced by Caires and Pfenning~\cite{Caires10}.
}

\condinc{


\begin{proof}[Theorem \ref{theo:subjred}]
 We reason by induction on 
the structure of $\ptone$. Since $\widehat{\ptone}=\procone\reds\proctwo$ the only possible last rules of
$\ptone$ can be: $\Lone, \Lbangb, \Lbangd,$, a linear cut ($\cut$) or an exponential cut ($\cutb$ or $\cutd$).
In all the other cases, the underlying process can only perform a visible action, as can be easily
verified by inspecting the rules from Figure~\ref{fig:pidslltr}. With this observation in mind, let
us inspect the operational semantics derivation proving that $\procone\reds\proctwo$. At some point
we will find two subprocesses of $\procone$, call them $\procthree$ and $\procfour$, which
communicate, causing an internal reduction. We here claim that this can only happen 
in presence of a cut, and only the communication between $\procthree$ and $\procfour$ must
occur along the channel involved in the cut. Now, it's only a matter of showing that the
just described situation can be ``resolved'' preserving types, and this can be done using the previous lemmas.
Some relevant case:

\begin{varitemize}

\item $\ptone = \Pcutb{\ptone_1}{\cone}{\ptone_2}$; assume $\conone = \conone_1, \conone_2$ and $\procone \cequ \rest{\cone}{\binc{\cone}{\cfour}{\para{\procone_1}{\procone_2}}}$. Now $\tyg{\conone_1}{\emcon}{\emcon}{\emcon}{\ptone_1 \rightsquigarrow \procone_1}{\cone : \typthree}$ and $\tyg{ \conone_2, \cone : \typone}{\contwo}{\conthree}{\ptone_2 \rightsquigarrow \procone_2}{\cthree : \typone}$ ,  by inversion; from $\procone \reds \proctwo$ either $\procone_2 \reds \proctwo_2$ and $\proctwo = \rest{\cone}{\binc{\cone}{\cfour}{\para{\procone_1}{\proctwo_2}}}$ or $\proctwo_2 \xrightarrow{\overline{\rest{\ctwo}{\outsc{\cone}{\ctwo}}}} \proctwo_2$ and $\rest{\cone}{\binc{\cone}{\cfour}{\para{\procone_1}{\rest{\ctwo}{\para{\procone_1}{\proctwo_2}}}}}$.

      First case:

      $\tyg{ \conone_2, \cone : \typone}{\contwo}{\conthree}{\pttwo_2 \rightsquigarrow \proctwo_2}{\cthree : \typone}$ for some $\pttwo_2$ with  $\ptone_2 \cpredequ\cred\cpredequ \pttwo_2$ by i.h.; $\Pcutb{\ptone_1}{\cone}{\ptone_2} \cpredequ\cred\cpredequ \Pcutb{\ptone_1}{\cone}{\pttwo_2}$ by congruence. Pick $E = \Pcutb{\ptone_1}{\cone}{\pttwo_2}$; then $\tyg{ \conone}{\contwo}{\conthree}{\pttwo \rightsquigarrow \proctwo}{\cthree : \typone}$ by $\cutb$.

      Second case:

      $\Pcutb{\ptone_1}{\cone}{\ptone_2} \cpredequ\cred\cpredequ \pttwo $ for some $\pttwo$; then $\tyg{\conone}{\contwo}{\conthree}{\pttwo \rightsquigarrow \procthree}{\cthree : \typone}$ for some $\procthree \cequ \proctwo$ by Corollary \ref{cor:cutb}.

\item $\ptone = \Pcutd{\ptone_1}{\cone}{\ptone_2}$. Now, $\procone \cequ \rest{\cone}{\binc{\cone}{\cfour}{\para{\procone_1}{\procone_2}}}$ and $\tyg{\contwo}{\emcon}{\emcon}{\emcon}{\ptone_1 \rightsquigarrow \procone_1}{\cone : \typthree}$, $\tyg{ \conone }{\contwo, \cone : \typone}{\conthree}{\ptone_2 \rightsquigarrow \procone_2}{\cthree : \typone}$ ,  by inversion; from $\procone \reds \proctwo$ either
$\procone_2 \reds \proctwo_2$ and $\proctwo = \rest{\cone}{\binc{\cone}{\cfour}{\para{\procone_1}{\proctwo_2}}}$ or $\proctwo_2 \xrightarrow{\overline{\rest{\ctwo}{\outsc{\cone}{\ctwo}}}} \proctwo_2$ and $\rest{\cone}{\binc{\cone}{\cfour}{\para{\procone_1}{\rest{\ctwo}{\para{\procone_1}{\proctwo_2}}}}}$

      First case:

      $\tyg{ \conone}{\contwo, \cone : \typone}{\conthree}{\pttwo_2 \rightsquigarrow \proctwo_2}{\cthree : \typone}$ for some $\pttwo_2$ with  $\ptone_2 \cpredequ\cred\cpredequ \pttwo_2$ by i.h. and $\Pcutd{\ptone_1}{\cone}{\ptone_2} \cpredequ\cred\cpredequ \Pcutd{\ptone_1}{\cone}{\pttwo_2}$ by congruence. Pick $E = \Pcutd{\ptone_1}{\cone}{\pttwo_2}$; then $\tyg{ \conone}{\contwo}{\conthree}{\pttwo \rightsquigarrow \proctwo}{\cthree : \typone}$ by $\cutd$

      Second case:

      $\Pcutd{\ptone_1}{\cone}{\ptone_2} \cpredequ\cred\cpredequ \pttwo $ for some $\pttwo$; then $\tyg{\conone}{\contwo}{\conthree}{\pttwo \rightsquigarrow \procthree}{\cthree : \typone}$ for some $\procthree \cequ \proctwo$ by Corollary \ref{cor:cutd}.
\end{varitemize}
This concludes the proof.
\end{proof}
}{}

\section{Proving Polynomial Bounds}\label{sect:boundint}
In this section, we prove the main result of this paper, namely some polynomial
bounds on the length of internal reduction sequences and on
the size of intermediate results for processes typable in \pidill. In other words, 
interaction will be shown to be bounded. The simplest formulation of this
result is the following:
\begin{theorem}\label{theo:boundint}
For every type $\typone$, there is a polynomial
$\polyone_\typone$ such that whenever
$\tyg{\emcon}{\emcon}{\cone:\typone}{\ptone}{\ctwo:\unit}$
and $\tyg{\emcon}{\emcon}{\emcon}{\pttwo}{\cone:\typone}$
where $\ptone$ and $\pttwo$ are normal and 
$\rest{\cone}{(\para{\pttotd{\ptone}}{\pttotd{\pttwo}})}\reds^\natone\procone$,
it holds that $\natone,\size{\procone}\leq\polyone_\typone(\size{\pttotd{\ptone}}+\size{\pttotd{\pttwo}})$
\end{theorem}
Intuitively, what Theorem~\ref{theo:boundint} says is that
the complexity of the interaction between two processes
typable without cuts and communicating through a channel with session
type $\typone$ is polynomial in their sizes, where the specific polynomial
involved only depends on $\typone$ itself. In other words, the complexity
of the interaction is not only bounded, but can be somehow ``read off'' from
the types of the communicating parties.

How does the proof of Theorem~\ref{theo:boundint} look like? Conceptually,
it can be thought of as being structured into four steps:
\begin{varenumerate}
\item\label{point:first}
   First of all, a natural number $\wei{\ptone}$ is attributed to any
   proof term $\ptone$. $\wei{\ptone}$ is said to be the \emph{weight}
   of $\ptone$.
\item
   Secondly, the weight of any proof term is shown to strictly decrease
   along computational reduction, not to increase along shifting
   reduction and to stay the same for equivalent proof terms.
\item\label{point:third}
   Thirdly, $\wei{\ptone}$ is shown to be bounded by a polynomial
   on $\size{\widehat{\ptone}}$, where the exponent only depends on the 
   nesting depth of boxes of $\ptone$, denoted $\bd{\ptone}$.
\item\label{point:fourth}
   Finally, the box depth $\bd{\ptone}$ of any proof term $\ptone$
   is shown to be ``readable'' from its type interface.
\end{varenumerate}
This is exactly what we are going to do in the rest of this section.
Please observe how points~\ref{point:first}--\ref{point:third} 
above allow to prove the following stronger result, from which Theorem~\ref{theo:boundint} 
easily follows, given point \ref{point:fourth}:
\begin{proposition}\label{prop:polysound}
For every $\natone\in\NN$, there is a polynomial 
$\polyone_\natone$ such that for every process $\procone$ 
with $\tyg{\conone}{\contwo}{\conthree}{\procone}{\tcone}$,
if $\procone\reds^\nattwo\proctwo$, then
$\nattwo,\size{\proctwo}\leq\polyone_{\bde{\procone}}(\size{\procone})$.
\end{proposition}

\subsection{Preliminary Definitions}
Some concepts have to be given before we can embark in the proof
of Proposition~\ref{prop:polysound}. First of all, we need to define what
the box-depth of a process is. Simply, given a process $\procone$, its
\emph{box-depth} $\bde{\procone}$ is the nesting-level of replications\footnote{This
terminology is derived from linear logic, where proofs obtained by
the promotion rule are usually called boxes} in $\procone$. As an example, the box-depth
of $\binc{\cone}{\ctwo}{\binc{\cthree}{\cfour}{\emproc}}$ is $2$, while
the one of $\rest{\cone}{\inwc{\ctwo}{\cthree}}$ is $0$. 
\condinc
{
Formally,
\begin{align*}
\bde{\PLone{\cone}{\ptone}}&=\bde{\ptone} &
\bde{\PRplusone{\ptone}}&=\bde{\ptone}\\
\bde{\PRone}&=0 &
\bde{\PRplustwo{\ptone}}&=\bde{\ptone}\\
\bde{\PLten{\cone}{\ctwo}{\cthree}{\ptone}}&=\bde{\ptone} &
\bde{\Pbemb{\cone}{\ctwo}{\ptone}}&=\bde{\ptone}\\
\bde{\PRten{\ptone}{\pttwo}}&=\max\{\bde{\ptone},\bde{\pttwo}\} &
\bde{\Pbemd{\cone}{\ctwo}{\ptone}}&=\bde{\ptone}\\
\bde{\PLlin{\cone}{\ptone}{\ctwo}{\pttwo}}&=\max\{\bde{\ptone},\bde{\pttwo}\} &
\bde{\PLbangb{\cone}{\ptone}}&=\bde{\ptone}\\
\bde{\PRlin{\cone}{\ptone}}&=\bde{\ptone} &
\bde{\PLbangd{\cone}{\ptone}}&=\bde{\ptone}\\
\bde{\PLwithone{\cone}{\ctwo}{\ptone}}&=\bde{\ptone} &
\bde{\PRbang{\cone_1,\ldots,\cone_n}{\ptone}}&=1+\bde{\ptone}\\
\bde{\PLwithtwo{\cone}{\ctwo}{\ptone}}&=\bde{\ptone} &
\bde{\Pcut{\ptone}{\cone}{\pttwo}}&=\max\{\bde{\ptone},\bde{\pttwo}\}\\
\bde{\PRwith{\ptone}{\pttwo}}&=\max\{\bde{\ptone},\bde{\pttwo}\} &
\bde{\Pcutb{\ptone}{\cone}{\pttwo}}&=\max\{\bde{\ptone}+1,\bde{\pttwo}\}\\
\bde{\PLplus{\cone}{\ctwo}{\ptone}{\cthree}{\pttwo}}&=\max\{\bde{\ptone},\bde{\pttwo}\} &
\bde{\Pcutd{\ptone}{\cone}{\pttwo}}&=\max\{\bde{\ptone}+1,\bde{\pttwo}\}
\end{align*}
}{}
Analogously, the
box-depth of a proof term $\ptone$ is simply $\bde{\widehat{\ptone}}$.

Now, suppose that $\tyg{\conone}{\contwo}{\conthree}{\ptone}{\tcone}$
and that $\cone:\typone$ belongs to either $\conone$ or $\contwo$, i.e.
that $\cone$ is an ``exponential'' channel in $\ptone$. A key parameter
is the \emph{virtual number of occurrences} of $\cone$ in $\ptone$, which is denoted
as $\foc{\cone}{\ptone}$. This parameter, as its name suggests, is not simply the number
of literal occurrences of $\cone$ in $\ptone$, but takes into account possible
duplications derived from cuts. So, for example, 
$\foc{\cfour}{\Pcutb{\ptone}{\cone}{\pttwo}}=\foc{\cone}{\pttwo}\cdot\foc{\cfour}{\ptone}+\foc{\cfour}{\pttwo}$,
while $\foc{\cfour}{\PRten{\ptone}{\pttwo}}$ is merely $\foc{\cfour}{\ptone}+\foc{\cfour}{\pttwo}$.
Obviously, $\foc{\cfour}{\Pbemb{\cone}{\cfour}{\ptone}}=1$ and
$\foc{\cfour}{\Pbemd{\cone}{\cfour}{\ptone}}=1$. 
\condinc
{
Formally:
\begin{align*}
\foc{\cfour}{\PLone{\cone}{\ptone}}&=\foc{\cfour}{\ptone}\\
\foc{\cfour}{\PRone}&=0\\
\foc{\cfour}{\PLten{\cone}{\ctwo}{\cthree}{\ptone}}&=\foc{\cfour}{\ptone}\\
\foc{\cfour}{\PRten{\ptone}{\pttwo}}&=\foc{\cfour}{\ptone}+\foc{\cfour}{\pttwo}\\
\foc{\cfour}{\PLlin{\cone}{\ptone}{\ctwo}{\pttwo}}&=\foc{\cfour}{\ptone}+\foc{\cfour}{\pttwo}\\
\foc{\cfour}{\PRlin{\cone}{\ptone}}&=\foc{\cfour}{\ptone}\\
\foc{\cfour}{\Pcut{\ptone}{\cone}{\pttwo}}&=\foc{\cfour}{\ptone}+\foc{\cfour}{\pttwo}\\
\foc{\cfour}{\Pcutb{\ptone}{\cone}{\pttwo}}&=\foc{\cone}{\pttwo}\cdot\foc{\cfour}{\ptone}+\foc{\cfour}{\pttwo}\\
\foc{\cfour}{\Pcutd{\ptone}{\cone}{\pttwo}}&=\foc{\cone}{\pttwo}\cdot\foc{\cfour}{\ptone}+\foc{\cfour}{\pttwo}\\
\foc{\cfour}{\Pbemb{\cone}{\cfour}{\ptone}}&=1\\
\foc{\cfour}{\Pbemd{\cone}{\cfour}{\ptone}}&=1\\
\foc{\cfour}{\Pbemb{\cone}{\ctwo}{\ptone}}&=0\\
\foc{\cfour}{\Pbemd{\cone}{\ctwo}{\ptone}}&=0\\
\foc{\cfour}{\PLbangb{\cone}{\ptone}}&=\foc{\cfour}{\ptone}\\
\foc{\cfour}{\PLbangd{\cone}{\ptone}}&=\foc{\cfour}{\ptone}\\
\foc{\cfour}{\PRbang{\cone_1,\ldots,\cone_n}{\ptone}}&=0\\
\foc{\cfour}{\PLplus{\cone}{\ctwo}{\ptone}{\cthree}{\pttwo}}&=\foc{\cfour}{\ptone}+\foc{\cfour}{\pttwo}\\
\foc{\cfour}{\PRplusone{\ptone}}&=\foc{\cfour}{\ptone}\\
\foc{\cfour}{\PRplustwo{\ptone}}&=\foc{\cfour}{\ptone}\\
\foc{\cfour}{\PLwithone{\cone}{\ctwo}{\ptone}}&=\foc{\cfour}{\ptone}\\
\foc{\cfour}{\PLwithtwo{\cone}{\ctwo}{\ptone}}&=\foc{\cfour}{\ptone}\\
\foc{\cfour}{\PRwith{\ptone}{\pttwo}}&=\foc{\cfour}{\ptone}+\foc{\cfour}{\pttwo}
\end{align*}
}{}

A channel in either the auxiliary or the exponential context can ``float'' to the linear context
as an effect of rules $\Lbangb$ or $\Lbangd$. From that moment on, it can only
be treated as a linear channel. As a consequence, it makes sense to
define the \emph{duplicability factor} of a proof term $\ptone$, written
$\dupf{\ptone}$, simply as the maximum of $\foc{\cone}{\ptone}$ over all instances 
of the rules $\Lbangb$ or $\Lbangd$ in $\ptone$, where $\cone$ is the involved
channel. For example, $\dupf{\PLbangb{\cone}{\ptone}}=\max\{\dupf{\ptone},\foc{\ctwo}{\ptone}\}$
and $\dupf{\PLlin{\cone}{\ptone}{\ctwo}{\pttwo}}=\max\{\dupf{\ptone},\dupf{\pttwo}\}$.
\condinc{
Formally, the duplicability factor $\dupf{\ptone}$ of $\ptone$ is defined as follows:
\begin{align*}
\dupf{\PLone{\cone}{\ptone}}&=\dupf{\ptone} &
\dupf{\PRplusone{\ptone}}&=\dupf{\ptone} \\
\dupf{\PRone}&=0 &
\dupf{\PRplustwo{\ptone}}&=\dupf{\ptone} \\
\dupf{\PLten{\cone}{\ctwo}{\cthree}{\ptone}}&=\dupf{\ptone} &
\dupf{\Pbemb{\cone}{\ctwo}{\ptone}}&=\dupf{\ptone} \\
\dupf{\PRten{\ptone}{\pttwo}}&=\max\{\dupf{\ptone},\dupf{\pttwo}\} &
\dupf{\Pbemd{\cone}{\ctwo}{\ptone}}&=\dupf{\ptone} \\
\dupf{\PLlin{\cone}{\ptone}{\ctwo}{\pttwo}}&=\max\{\dupf{\ptone},\dupf{\pttwo}\} &
\dupf{\PLbangb{\cone}{\ptone}}&=\max\{\dupf{\ptone},\foc{\ctwo}{\ptone}\} \\
\dupf{\PRlin{\cone}{\ptone}}&=\dupf{\ptone} &
\dupf{\PLbangd{\cone}{\ptone}}&=\max\{\dupf{\ptone},\foc{\ctwo}{\ptone}\} \\
\dupf{\PLwithone{\cone}{\ctwo}{\ptone}}&=\dupf{\ptone} &
\dupf{\PRbang{\cone_1,\ldots,\cone_n}{\ptone}}&=\dupf{\ptone} \\
\dupf{\PLwithtwo{\cone}{\ctwo}{\ptone}}&=\dupf{\ptone} &
\dupf{\Pcut{\ptone}{\cone}{\pttwo}}&=\max\{\dupf{\ptone},\dupf{\pttwo}\} \\
\dupf{\PRwith{\ptone}{\pttwo}}&=\max\{\dupf{\ptone},\dupf{\pttwo}\} & 
\dupf{\Pcutb{\ptone}{\cone}{\pttwo}}&=\max\{\dupf{\ptone},\dupf{\pttwo}\}\\
\dupf{\PLplus{\cone}{\ctwo}{\ptone}{\cthree}{\pttwo}}&=\max\{\dupf{\ptone},\dupf{\pttwo}\} & 
\dupf{\Pcutd{\ptone}{\cone}{\pttwo}}&=\max\{\dupf{\ptone},\dupf{\pttwo}\}
\end{align*}}{}

It's now possible to give the definition of $\wei{\ptone}$, namely the \emph{weight}
of the proof term $\ptone$. Before doing that, however, it is necessary to give
a parameterized notion of weight, denoted $\weip{\natone}{\ptone}$. Intuitively,
$\weip{\natone}{\ptone}$ is defined similarly to $\size{\widehat{\ptone}}$.
However, every input and output action in $\widehat{\ptone}$ can possibly count for
more than one:
\begin{varitemize}
\item
   Everything inside $\ptone$ in $\PRbang{\cone_1,\ldots,\cone_n}{\ptone}$
   counts for $\natone$;
\item
   Everything inside $\ptone$ in either 
   $\Pcutb{\ptone}{\cone}{\pttwo}$ or 
   $\Pcutd{\ptone}{\cone}{\pttwo}$ counts for 
   $\foc{\cone}{\pttwo}$.
\end{varitemize}
For example, 
$\weip{\natone}{\Pcutd{\ptone}{\cone}{\pttwo}}=\foc{\cone}{\pttwo}\cdot\weip{\natone}{\ptone}+\weip{\natone}{\pttwo}$,
while $\weip{\natone}{\PLwithtwo{\cone}{\ctwo}{\ptone}}=1+\weip{\natone}{\ptone}$.
\condinc{
Formally:
\begin{align*}
\weip{\natone}{\PLone{\cone}{\ptone}}&=\weip{\natone}{\ptone}\\
\weip{\natone}{\PRone}&=0\\
\weip{\natone}{\PLten{\cone}{\ctwo}{\cthree}{\ptone}}&=1+\weip{\natone}{\ptone}\\
\weip{\natone}{\PRten{\ptone}{\pttwo}}&=1+\weip{\natone}{\ptone}+\weip{\natone}{\pttwo}\\
\weip{\natone}{\PLlin{\cone}{\ptone}{\ctwo}{\pttwo}}&=1+\weip{\natone}{\ptone}+\weip{\natone}{\pttwo}\\
\weip{\natone}{\PRlin{\cone}{\ptone}}&=1+\weip{\natone}{\ptone}\\
\weip{\natone}{\Pcut{\ptone}{\cone}{\pttwo}}&=\weip{\natone}{\ptone}+\weip{\natone}{\pttwo}\\
\weip{\natone}{\Pcutb{\ptone}{\cone}{\pttwo}}&=\foc{\cone}{\pttwo}\cdot\weip{\natone}{\ptone}+\weip{\natone}{\pttwo}\\
\weip{\natone}{\Pcutd{\ptone}{\cone}{\pttwo}}&=\foc{\cone}{\pttwo}\cdot\weip{\natone}{\ptone}+\weip{\natone}{\pttwo}\\
\weip{\natone}{\Pbemb{\cone}{\ctwo}{\ptone}}&=1+\weip{\natone}{\ptone}\\
\weip{\natone}{\Pbemd{\cone}{\ctwo}{\ptone}}&=1+\weip{\natone}{\ptone}\\
\weip{\natone}{\PLbangb{\cone}{\ptone}}&=\weip{\natone}{\ptone}\\
\weip{\natone}{\PLbangd{\cone}{\ptone}}&=\weip{\natone}{\ptone}\\
\weip{\natone}{\PRbang{\cone_1,\ldots,\cone_n}{\ptone}}&=\natone\cdot(\weip{\natone}{\ptone}+1)\\
\weip{\natone}{\PLplus{\cone}{\ctwo}{\ptone}{\cthree}{\pttwo}}&=1+\weip{\natone}{\ptone}+\weip{\natone}{\pttwo}\\
\weip{\natone}{\PRplusone{\ptone}}&=1+\weip{\natone}{\ptone}\\
\weip{\natone}{\PRplustwo{\ptone}}&=1+\weip{\natone}{\ptone}\\
\weip{\natone}{\PLwithone{\cone}{\ctwo}{\ptone}}&=1+\weip{\natone}{\ptone}\\
\weip{\natone}{\PLwithtwo{\cone}{\ctwo}{\ptone}}&=1+\weip{\natone}{\ptone}\\
\weip{\natone}{\PRwith{\ptone}{\pttwo}}&=1+\weip{\natone}{\ptone}+\weip{\natone}{\pttwo}
\end{align*}
}{}
Now, $\wei{\ptone}$ is simply $\weip{\dupf{\ptone}}{\ptone}$.
\condinc{}
{The concepts we have just introduced are more precisely defined in~\cite{DalLagoDiGiambe2011ev}.}
\subsection{Monotonicity Results}
The crucial ingredient for proving polynomial bounds are a series of results about how
the weight $\ptone$ evolves when $\ptone$ is put in relation with another proof term
$\pttwo$ by way of either $\cred$, $\cpred$ or $\cequ$.
\condinc{
\begin{lemma}
For every $\ptone$, $\dupf{\ptone}=\dupf{\lift{\ptone}}$ and
for every $\natone$, $\weip{\natone}{\ptone}=\weip{\natone}{\lift{\ptone}}$.
\end{lemma}
}{}
Whenever a proof term $\ptone$ computationally reduces to $\pttwo$, the underlying
weight is guaranteed to strictly decrease:
\begin{proposition}\label{prop:creddecr}
If $\tyg{\conone}{\contwo}{\conthree}{\ptone}{\tcone}$ and
$\ptone\cred\pttwo$, then $\tyg{\confour}{\confive}{\conthree}{\pttwo}{\tcone}$
(where $\conone,\contwo=\confour,\confive$), $\dupf{\pttwo}\leq\dupf{\ptone}$ and 
$\wei{\pttwo}<\wei{\ptone}$.
\end{proposition}
\begin{proof}
By induction on the proof that $\ptone\cred\pttwo$. Some interesting cases:
\begin{varitemize}
\item
  Suppose that 
  $
   \ptone=\Pcut{\PRlin{\ctwo}{\ptthree}}{\cone}{\PLlin{\cone}{\ptfour}{\cone}{\ptfive}}\cred\Pcut{\Pcut{\ptfour}{\ctwo}{\ptthree}}{\cone}{\ptfive}=\pttwo
  $.
  Then,
  \begin{align*}
    \dupf{\ptone}&=\max\{\dupf{\ptthree},\dupf{\ptfour},\dupf{\ptfive}\}=\dupf{\pttwo};\\
    \wei{\ptone}&=\weip{\dupf{\ptone}}{\ptone}=3+\weip{\dupf{\ptone}}{\ptthree}+\weip{\dupf{\ptone}}{\ptfour}+
          \weip{\dupf{\ptone}}{\ptfive}\\
                &>2+\weip{\dupf{\pttwo}}{\ptthree}+\weip{\dupf{\pttwo}}{\ptfour}+\weip{\dupf{\pttwo}}{\ptfive}
                 =\weip{\dupf{\pttwo}}{\pttwo}=\wei{\pttwo}.
  \end{align*}
\item
  Suppose that
  $
  \ptone=\Pcut{\PRwith{\ptthree}{\ptfour}}{\cone}{\PLwithone{\cone}{\ctwo}{\ptfive}}\cred\Pcut{\ptthree}{\cone}{\ptfive}=\pttwo
  $.
  Then,
  \begin{align*}
    \dupf{\ptone}&=\max\{\dupf{\ptthree},\dupf{\ptfour},\dupf{\ptfive}\}=\dupf{\pttwo};\\
    \wei{\ptone}&=\weip{\dupf{\ptone}}{\ptone}=3+\weip{\dupf{\ptone}}{\ptthree}+\weip{\dupf{\ptone}}{\ptfour}+
           \weip{\dupf{\ptone}}{\ptfive}\\
                &>2+\weip{\dupf{\pttwo}}{\ptthree}+\weip{\dupf{\pttwo}}{\ptfour}+\weip{\dupf{\pttwo}}{\ptfive}
                =\weip{\dupf{\pttwo}}{\pttwo}=\wei{\pttwo}.
  \end{align*}
\item
  Suppose that
    $
    \ptone=\Pcutb{\ptthree}{\cone}{\Pbemb{\cone}{\ctwo}{\ptfour}}\cred\Pcut{\lift{\ptthree}}{\ctwo}{\Pcutd{\ptthree}{\cone}{\lift{\ptfour}}}=\pttwo
    $.  
  Then, 
  \begin{align*}
    \dupf{\ptone}&=\max\{\dupf{\lift{\ptthree}},\dupf{\lift{\ptfour}}\}
        =\max\{\dupf{\ptthree},\dupf{\ptthree},\dupf{\ptfour}\}=\dupf{\pttwo};\\
    \wei{\ptone}&=\weip{\dupf{\ptone}}{\ptone}=\foc{\cone}{\Pbemb{\cone}{\ctwo}{\ptfour}}\cdot
                \weip{\dupf{\ptone}}{\lift{\ptthree}}+\weip{\dupf{\ptone}}{\Pbemb{\cone}{\ctwo}{\ptfour}}\\
                &=\weip{\dupf{\ptone}}{\ptthree}+\weip{\dupf{\ptone}}{\Pbemb{\cone}{\ctwo}{\ptfour}}
                =\weip{\dupf{\ptone}}{\ptthree}+1+\weip{\dupf{\ptone}}{\ptfour}\\
                &\geq\weip{\dupf{\pttwo}}{\ptthree}+1+\weip{\dupf{\pttwo}}{\ptfour}\\
                &>\weip{\dupf{\pttwo}}{\ptthree}+\weip{\dupf{\pttwo}}{\ptfour}
                =\weip{\dupf{\pttwo}}{\ptthree}+0\cdot\weip{\dupf{\pttwo}}{\ptthree}+\weip{\dupf{\pttwo}}{\ptfour}\\
                &=\weip{\dupf{\pttwo}}{\ptthree}+\foc{\cone}{\ptfour}\cdot\weip{\dupf{\pttwo}}{\ptthree}+\weip{\dupf{\pttwo}}{\ptfour}\\
                &=\weip{\dupf{\pttwo}}{\pttwo}=\wei{\pttwo}.
  \end{align*}
\item
  Suppose that
    $$
      \ptone=\Pcutd{\ptthree}{\cone}{\Pbemd{\cone}{\ctwo}{\ptfour}}\cred
        \Pcut{\lift{\ptthree}}{\ctwo}{\Pcutd{\ptthree}{\cone}{\ptfour}}=\pttwo.
    $$
  Then we can proceed exactly as in the previous case.
\end{varitemize}
This concludes the proof.
\end{proof}
Shift reduction, on the other hand, is \emph{not} guaranteed to induce a strict decrease on the underlying
weight which, however, cannot increase:
\begin{proposition}\label{prop:cpreddecr}
If $\tyg{\conone}{\contwo}{\conthree}{\ptone}{\tcone}$ and
$\ptone\cpred\pttwo$, then $\tyg{\conone}{\contwo}{\conthree}{\pttwo}{\tcone}$,
$\dupf{\pttwo}\leq\dupf{\ptone}$ and $\wei{\pttwo}\leq\wei{\ptone}$.
\end{proposition}
\begin{proof}
By induction on the proof that $\ptone\cpred\pttwo$. Some interesting cases:
\begin{varitemize}
\item
  Suppose that
    $$
      \ptone=\Pcut{\PRbang{\cone_1,\ldots,\cone_n}{\ptthree}}{\cone}{\PLbangb{\cone}{\ptfour}}\cpred
          \PLbangb{\cone_1}{\PLbangb{\cone_2}{\ldots\PLbangb{\cone_n}{\Pcutb{\ptthree}{\ctwo}{\ptfour}}}}=\pttwo.
    $$
  Then,
  \begin{align*}
    \dupf{\ptone}&=\max\{\dupf{\ptthree},\dupf{\ptfour}\}=\dupf{\pttwo}\\
    \wei{\ptone}&=\weip{\dupf{\ptone}}{\ptone}=\dupf{\ptone}\cdot\weip{\dupf{\ptone}}{\ptthree}+
                  \weip{\dupf{\ptone}}{\ptfour}\geq\foc{\ctwo}{\ptfour}\cdot\weip{\dupf{\ptone}}{\ptthree}+\weip{\dupf{\ptone}}{\ptfour}\\
                &=\foc{\ctwo}{\ptfour}\cdot\weip{\dupf{\pttwo}}{\ptthree}+\weip{\dupf{\pttwo}}{\ptfour}=\weip{\dupf{\pttwo}}{\pttwo}=\wei{\pttwo}.
  \end{align*}
\item
  Suppose that 
    $$
      \ptone=\Pcut{\PRbang{\cone_1,\ldots,\cone_n}{\ptthree}}{\cone}{\PLbangd{\cone}{\ptfour}}\cpred
          \PLbangd{\cone_1}{\PLbangd{\cone_2}{\ldots\PLbangd{\cone_n}{\Pcutd{\ptthree}{\ctwo}{\ptfour}}}}=\pttwo.
    $$
  Then we can proceed as in the previous case.
\end{varitemize}
This concludes the proof.
\end{proof}
Finally, equivalence leaves the weight unchanged:
\begin{proposition}\label{prop:cequun}
If $\tyg{\conone}{\contwo}{\conthree}{\ptone}{\tcone}$ and
$\ptone\cequ\pttwo$, then $\tyg{\conone}{\contwo}{\conthree}{\pttwo}{\tcone}$,
$\dupf{\pttwo}=\dupf{\ptone}$ and $\wei{\pttwo}=\wei{\ptone}$.
\end{proposition}
\begin{proof}
By induction on the proof that $\ptone\cequ\pttwo$. Some interesting cases:
\begin{varitemize}
\item
  Suppose that
  $$
  \ptone=\Pcut{\ptthree}{\cone}{\Pcut{\ptfour_\cone}{\ctwo}{\ptfive_\ctwo}}\cequ\Pcut{\Pcut{\ptthree}{\cone}{\ptfour_\cone}}{\ctwo}{\ptfive_{\ctwo}}=\pttwo.
  $$  
  Then:
  \begin{align*}
    \dupf{\ptone}&=\max\{\dupf{\ptthree},\dupf{\ptfour_\cone},\dupf{\ptfive_\ctwo}\}=\dupf{\pttwo}\\
    \wei{\ptone}&=\weip{\dupf{\ptone}}{\ptone}=\weip{\dupf{\ptone}}{\ptthree}+\weip{\dupf{\ptone}}{\ptfour_\cone}+\weip{\dupf{\ptone}}{\ptfive_\ctwo}\\
        &=\weip{\dupf{\pttwo}}{\ptthree}+\weip{\dupf{\pttwo}}{\ptfour_\cone}+\weip{\dupf{\pttwo}}{\ptfive_\ctwo}=\weip{\dupf{\pttwo}}{\pttwo}=\wei{\pttwo}.
  \end{align*}
\item 
  Suppose that
  $$
  \ptone=\Pcut{\ptthree}{\cone}{\Pcut{\ptfour}{\ctwo}{\ptfive_{\cone\ctwo}}}\cequ\Pcut{\ptfour}{\cone}{\Pcut{\ptthree}{\ctwo}{\ptfive_{\cone\ctwo}}}=\pttwo.
  $$  
  Then we can proceed as in the previous case.
\item
  Suppose that
  $$
  \ptone=\Pcut{\ptthree}{\cone}{\Pcutb{\ptfour}{\ctwo}{\ptfive_{\cone\ctwo}}}\cequ\Pcutb{\ptfour}{\ctwo}{\Pcut{\ptthree}{\cone}{\ptfive_{\cone\ctwo}}}=\pttwo.
  $$
  Then, since $\foc{\ctwo}{\ptthree}=0$,
  \begin{align*}
    \dupf{\ptone}&=\max\{\dupf{\ptthree},\dupf{\ptfour},\dupf{\ptfive_{\cone\ctwo}}\}=\dupf{\pttwo}\\
    \wei{\ptone}&=\weip{\dupf{\ptone}}{\ptone}=\weip{\dupf{\ptone}}{\ptthree}+\foc{\ctwo}{\ptfive_{\cone\ctwo}}
                   \cdot\weip{\dupf{\ptone}}{\ptfour}+\weip{\dupf{\ptone}}{\ptfive_{\cone\ctwo}}\\
        &=\weip{\dupf{\ptone}}{\ptthree}+\foc{\ctwo}{\Pcut{\ptthree}{\cone}{\ptfive_{\cone\ctwo}}}
                   \cdot\weip{\dupf{\ptone}}{\ptfour}+\weip{\dupf{\ptone}}{\ptfive_{\cone\ctwo}}\\
        &=\weip{\dupf{\pttwo}}{\ptthree}+\foc{\ctwo}{\Pcut{\ptthree}{\cone}{\ptfive_{\cone\ctwo}}}
                   \cdot\weip{\dupf{\pttwo}}{\ptfour}+\weip{\dupf{\pttwo}}{\ptfive_{\cone\ctwo}}\\
        &=\weip{\dupf{\pttwo}}{\pttwo}=\wei{\pttwo}.
  \end{align*}
\item
  Suppose that
  $$
  \ptone=\Pcutd{\ptthree}{\cone}{\Pcut{\ptfour_{\cone}}{\ctwo}{\ptfive_{\cone \ctwo}}}\cequ
    \Pcut{\Pcutd{\ptthree}{\cone}{\ptfour_{\cone}}}{\ctwo}{\Pcutd{\ptthree}{\cone}{\ptfive_{\cone \ctwo}}}=\pttwo.
  $$
  Then,
  \begin{align*}
    \dupf{\ptone}&=\max\{\dupf{\ptthree},\dupf{\ptfour_\cone},\dupf{\ptfive_{\cone\ctwo}}\}=\dupf{\pttwo}\\
    \wei{\ptone}&=\foc{\cone}{\Pcut{\ptfour_{\cone}}{\ctwo}{\ptfive_{\cone \ctwo}}}\cdot\weip{\dupf{\ptone}}{\ptthree}+
      \weip{\dupf{\ptone}}{\ptfour_\cone}+\weip{\dupf{\ptone}}{\ptfive_{\cone\ctwo}}\\
       &=(\foc{\cone}{\ptfour_{\cone}}+\foc{\cone}{\ptfive_{\cone\ctwo}})\cdot\weip{\dupf{\ptone}}{\ptthree}+
      \weip{\dupf{\ptone}}{\ptfour_\cone}+\weip{\dupf{\ptone}}{\ptfive_{\cone\ctwo}}\\
       &=(\foc{\cone}{\ptfour_{\cone}}\cdot\weip{\dupf{\ptone}}{\ptthree}+\foc{\cone}{\ptfive_{\cone\ctwo}})\cdot
        \weip{\dupf{\ptone}}{\ptthree}+\weip{\dupf{\ptone}}{\ptfour_\cone}+\weip{\dupf{\ptone}}{\ptfive_{\cone\ctwo}}\\
       &=\weip{\dupf{\ptone}}{\Pcutd{\ptthree}{\cone}{\ptfour_{\cone}}}+\weip{\dupf{\ptone}}{\Pcutd{\ptthree}{\cone}{\ptfive_{\cone \ctwo}}}\\
       &=\weip{\dupf{\ptone}}{\pttwo}=\weip{\dupf{\pttwo}}{\pttwo}=\wei{\pttwo}.
  \end{align*}
\end{varitemize}
This concludes the proof.
\end{proof}
Now, consider again the subject reduction theorem (Theorem~\ref{theo:subjred}): what it guarantees is that
whenever $\procone\reds\proctwo$ and $\widehat{\ptone}=\procone$, there is $\pttwo$ with
$\widehat{\pttwo}=\proctwo$ and $\ptone\cpredequ\cred\cpredequ\pttwo$. In view of the three propositions
we have just stated and proved, it's clear that $\wei{\ptone}>\pttwo$. Altogether, this implies that
$\wei{\ptone}$ is an upper bound on the number or internal reduction steps $\widehat{\ptone}$ can
perform. But is $\wei{\ptone}$ itself bounded?
\subsection{Bounding the Weight}
What kind of bounds can we expect to prove for $\wei{\ptone}$?
More specifically, how related are $\wei{\ptone}$ and $\size{\widehat{\ptone}}$?
\condinc{
\begin{lemma}\label{lemma:foc}
Suppose $\tyg{\conone}{\contwo}{\conthree}{\ptone}{\tcone}$. Then
\begin{varenumerate}
\item
  If $\cone\in\conone$, then $\foc{\cone}{\ptone}\leq 1$;
\item
  If $\cone\in\contwo$, then $\foc{\cone}{\ptone}\leq\size{\ptone}$;
\item
  If $\cone\in\conthree$, then $\foc{\cone}{\ptone}=0$;
\end{varenumerate}
\end{lemma}
\begin{proof}
By induction on the structure of a type derivation $\tdone$
for $\tyg{\conone}{\contwo}{\conthree}{\ptone}{\tcone}$. Some
interesting cases:
\begin{varitemize}
\item
  If $\tdone$ is
  $$
  \infer[\Rten]
  {\tyg{\conone_1,\conone_2}{\contwo}{\conthree_1,\conthree_2}{\PRten{\ptone_1}{\ptone_2}}{\ctwo: A \tens B}}
  {\tdtwo_1:\tyg{\conone_1}{\contwo}{\conthree_1}{\ptone_1}{\cthree:A} & \tdtwo_2:\tyg{\conone_2}{\contwo}{\conthree_2}{\ptone_2}{\ctwo:B}}
  $$
  then
  \begin{align*}
    \foc{\cone}{\PRten{\ptone_1}{\ptone_2}}&=\foc{\cone}{\ptone_1}\leq 1&\mbox{if }\cone\in\conone_1\\
    \foc{\cone}{\PRten{\ptone_1}{\ptone_2}}&=\foc{\cone}{\ptone_2}\leq 1&\mbox{if }\cone\in\conone_2\\
    \foc{\cone}{\PRten{\ptone_1}{\ptone_2}}&=\foc{\cone}{\ptone_1}+\foc{\cone}{\ptone_1}\\
       &\leq\size{\ptone_1}+\size{\ptone_2}\leq\size{\PRten{\ptone_1}{\ptone_2}}&\mbox{if }\cone\in\contwo\\
    \foc{\cone}{\PRten{\ptone_1}{\ptone_2}}&=\foc{\cone}{\ptone_1}=0&\mbox{if }\cone\in\conthree_1\\
    \foc{\cone}{\PRten{\ptone_1}{\ptone_2}}&=\foc{\cone}{\ptone_2}=0&\mbox{if }\cone\in\conthree_2\\
  \end{align*}
\item
  If $\tdone$ is
  $$
  \infer[\cutd]
  {\tyg{\conone_2}{\conone_1}{\conthree}{\Pcutd{\ptone_1}{\ctwo}{\ptone_2}}{\tcone}}
  {\tyg{\conone_1}{\emcon}{\emcon}{\emcon}{\ptone_1}{\cthree:A} & \tyg{\conone_2}{\conone_1,\ctwo:A}{\conthree}{\ptone_2}{\tcone}}
  $$
  then:
  \begin{align*}
    \foc{\cone}{\Pcutd{\ptone_1}{\ctwo}{\ptone_2}}&=\foc{\ctwo}{\ptone_2}\cdot\foc{\cone}{\ptone_1}+\foc{\cone}{\ptone_2}\\
       &\leq\size{\ptone_2}\cdot 1+\size{\ptone_1}\leq\size{\Pcutd{\ptone_1}{\ctwo}{\pttwo_2}}&\mbox{if }\cone\in\conone_1\\
    \foc{\cone}{\Pcutd{\ptone_1}{\ctwo}{\ptone_2}}&=\foc{\ctwo}{\ptone_2}\cdot\foc{\cone}{\ptone_1}+\foc{\cone}{\ptone_2}\\
       &\leq\size{\ptone_2}\cdot 0+1=1&\mbox{if }\cone\in\conone_2\\
    \foc{\cone}{\Pcutd{\ptone_1}{\ctwo}{\pttwo_2}}&=\foc{\ctwo}{\ptone_2}\cdot\foc{\cone}{\ptone_1}+\foc{\cone}{\ptone_2}\\
       &\leq\size{\ptone_2}\cdot 0+1=1&\mbox{if }\cone\in\conthree\\
  \end{align*}
\item
  If $\tdone$ is
  $$
  \infer[\cutw]
  {\tyg{\conone_2}{\contwo}{\conthree}{\Pcutw{\ptone_1}{\ctwo}{\ptone_2}}{\tcone}}
  {\tyg{\conone_1}{\emcon}{\emcon}{\emcon}{\ptone_1}{\cthree:A} & \tyg{\conone_2}{\contwo}{\conthree}{\ptone_2}{\tcone}}
  $$  
  then:
  \begin{align*}
    \foc{\cone}{\Pcutw{\ptone_1}{\ctwo}{\ptone_2}}&=\foc{\ctwo}{\ptone_2}\cdot\foc{\cone}{\ptone_1}+\foc{\cone}{\ptone_2}\\
       &\leq 0\cdot 1+0=0&\mbox{if }\cone\in\conone_1\\
    \foc{\cone}{\Pcutw{\ptone_1}{\ctwo}{\ptone_2}}&=\foc{\ctwo}{\ptone_2}\cdot\foc{\cone}{\ptone_1}+\foc{\cone}{\ptone_2}\\
       &\leq 0\cdot 0+1=1&\mbox{if }\cone\in\conone_2\\
    \foc{\cone}{\Pcutw{\ptone_1}{\ctwo}{\pttwo_2}}&=\foc{\ctwo}{\ptone_2}\cdot\foc{\cone}{\ptone_1}+\foc{\cone}{\ptone_2}\\
       &\leq 0\cdot 0+\size{\ptone_2}\leq\size{\Pcutd{\ptone_1}{\ctwo}{\pttwo_2}}&\mbox{if }\cone\in\contwo\\
    \foc{\cone}{\Pcutw{\ptone_1}{\ctwo}{\pttwo_2}}&=\foc{\ctwo}{\ptone_2}\cdot\foc{\cone}{\ptone_1}+\foc{\cone}{\ptone_2}\\
       &\leq 0\cdot 0+1=1&\mbox{if }\cone\in\conthree\\
  \end{align*}
\end{varitemize}
This concludes the proof.
\end{proof}}{}

\begin{lemma}\label{lemma:dupf}
Suppose $\tyg{\conone}{\contwo}{\conthree}{\ptone}{\tcone}$. Then
$\dupf{\ptone}\leq\size{\ptone}$.
\end{lemma}
\begin{proof}
An easy induction on the structure of a type derivation $\tdone$ for 
$\tyg{\conone}{\contwo}{\conthree}{\ptone}{\tcone}$. 
\condinc{Some interesting cases:
\begin{varitemize}
\item
  If $\tdone$ is
  $$
  \infer[\cutd]
  {\tyg{\conone_2}{\contwo,\conone_1}{\conthree}{\Pcutd{\ptone_1}{\ctwo}{\ptone_2}}{\tcone}}
  {\tyg{\conone_1}{\emcon}{\emcon}{\emcon}{\ptone_1}{\cthree:A} & \tyg{\conone_2}{\contwo,\ctwo:A}{\conthree}{\ptone_2}{\tcone}}
  $$
  then, by Lemma~\ref{lemma:foc} and by induction hypothesis:
  \begin{align*}
    \dupf{\Pcutd{\ptone_1}{\ctwo}{\ptone_2}}&=\max\{\dupf{\ptone_1},\dupf{\ptone_2}\}\\
       &\leq\max\{\size{\ptone_1},\size{\ptone_2}\}\\
       &\leq\size{\Pcutd{\ptone_1}{\ctwo}{\ptone_2}}
  \end{align*}
\end{varitemize}
This concludes the proof.}{}
\end{proof}

\begin{lemma}\label{lemma:weibound}
If $\tyg{\conone}{\contwo}{\conthree}{\ptone}{\tcone}$, then for every
$\natone\geq\dupf{\ptone}$, $\weip{\natone}{\ptone}\leq\size{\widehat{\ptone}}\cdot\natone^{\bde{\widehat{\ptone}}+1}$.
\end{lemma}
\begin{proof}
By induction on the structure of $\ptone$. Some interesting cases:
\begin{varitemize}
\condinc{\item
  If $\ptone=\PRten{\pttwo}{\ptthree}$, then:
  \begin{align*}
    \weip{\natone}{\PRten{\pttwo}{\ptthree}}&=1+\weip{\natone}{\pttwo}+\weip{\natone}{\ptthree}\\
      &\leq 1+\size{\pttwo}\cdot\natone^{\bde{\pttwo}+1}+\size{\ptthree}\cdot\natone^{\bde{\ptthree}+1}\\
      &\leq 1+(\size{\pttwo}+\size{\ptthree})\cdot\natone^{\max\{\bde{\pttwo}+1,\bde{\ptthree}+1\}}\\
      &\leq (1+\size{\pttwo}+\size{\ptthree})\cdot\natone^{\max\{\bde{\pttwo}+1,\bde{\ptthree}+1\}}\\
      &\leq \size{\PRten{\pttwo}{\ptthree}}\cdot\natone^{\bde{\PRten{\pttwo}{\ptthree}}+1}\\
  \end{align*}}{}
\item
  If $\ptone=\Pcutb{\ptone}{\cone}{\pttwo}$, then:
  \begin{align*}
    \weip{\natone}{\Pcutb{\ptone}{\cone}{\pttwo}}&=\foc{\cone}{\pttwo}\cdot(\weip{\natone}{\ptone}+1)+\weip{\natone}{\pttwo}\\
      &\leq \foc{\cone}{\pttwo}\cdot(\size{\ptone}\cdot\natone^{\bde{\ptone}+1}+1)+\size{\pttwo}\cdot\natone^{\bde{\pttwo}+1}\\
      &\leq \natone\cdot\size{\ptone}\cdot\natone^{\bde{\ptone}+1}+\natone+\size{\pttwo}\cdot\natone^{\bde{\pttwo}+1}\\
      &\leq \size{\ptone}\cdot\natone^{\bde{\ptone}+2}+\natone^{\bde{\pttwo}+1}+\size{\pttwo}\cdot\natone^{\bde{\pttwo}+1}\\
      &\leq (\size{\ptone}+\size{\pttwo}+1)\cdot\natone^{\max\{\bde{\ptone}+2,\bde{\pttwo}+1\}}\\
      &=\size{\Pcutb{\ptone}{\cone}{\pttwo}}\cdot\natone^{\bde{\Pcutb{\ptone}{\cone}{\pttwo}}}.\\
  \end{align*}
\item
  If $\ptone=\PRbang{\cone_1,\ldots,\cone_n}{\pttwo}$, then:
  \begin{align*}
    \weip{\natone}{\PRbang{\cone_1,\ldots,\cone_n}{\pttwo}}&=\natone\cdot(\weip{\natone}{\pttwo}+1)\\
       &\leq \natone\cdot\size{\pttwo}\cdot\natone^{\bde{\pttwo}+1}+\natone\\
       &\leq\size{\pttwo}\cdot\natone^{\bde{\pttwo}+2}+\natone^{\bde{\pttwo}+2}\\
       &=(1+\size{\pttwo})\cdot\natone^{\bde{\PRbang{\cone_1,\ldots,\cone_n}{\pttwo}}+1}\\
       &=\size{\PRbang{\cone_1,\ldots,\cone_n}{\pttwo}}\cdot\natone^{\bde{\PRbang{\cone_1,\ldots,\cone_n}{\pttwo}}+1}.
  \end{align*}
\end{varitemize}
This concludes the proof.
\end{proof}
\subsection{Putting Everything Together}
We now have almost all the necessary ingredients to obtain a proof of Proposition~\ref{prop:polysound}: the only
missing tales are the bounds on the size of any reducts, since the polynomial bounds on the length of internal
reductions are exactly the ones from Lemma~\ref{lemma:weibound}. Observe, however, that the latter
induces the former:
\begin{lemma}\label{lemma:spacevstime}
Suppose that $\procone\reds^\natone\proctwo$. Then
$\size{\proctwo}\leq\natone\cdot\size{\procone}$.
\end{lemma}
\begin{proof}
By induction on $\natone$, enriching the statement as follows:
whenever $\procone\reds^\natone\proctwo$, both
$\size{\proctwo}\leq\natone\cdot\size{\procone}$ and
$\size{\procthree}\leq\size{\procone}$ for every
subprocess $\procthree$ of $\proctwo$ in the form 
$\binc{\cone}{\ctwo}{\procfour}$. 
\end{proof}
\condinc
{
\begin{lemma}\label{lemma:procvspt}
For every $\ptone$, $\bde{\ptone}=\bde{\widehat{\ptone}}$ and
$\size{\ptone}=\size{\widehat{\ptone}}$.
\end{lemma}
Finally:
\begin{proof}[Proposition~\ref{prop:polysound}]
Let $\{\polytwo_\natone\}_{\natone\in\NN}$ the polynomials coming
from Lemma~\ref{lemma:weibound}. The polynomials we are looking
for are defined as follows:
$$
\polyone_\natone(\indone)=\polytwo_\natone(\indone)+\indone\cdot\polytwo_\natone(\indone).
$$
Now, suppose that $\procone\reds^\nattwo\proctwo$. By Theorem~\ref{theo:subjred},
there are proof terms $\ptone,\pttwo$ such that
$\procone=\widehat{\ptone}$, $\proctwo=\widehat{\pttwo}$ and
$$
\ptone(\cpredequ\cred\cpredequ)^\nattwo\pttwo.
$$
Now, from propositions \ref{prop:creddecr}, \ref{prop:cpreddecr} and \ref{prop:cequun}, it follows that
$$
\wei{\ptone}\geq \nattwo+\wei{\pttwo}\geq \nattwo.
$$
As a consequence, by Lemma~\ref{lemma:weibound} and Lemma~\ref{lemma:procvspt}, 
$$
\nattwo\leq\polytwo_{\bde{\ptone}}(\size{\ptone})\leq\polytwo_{\bde{\procone}}(\size{\procone})\leq\polyone_{\bde{\procone}}(\size{\procone}).
$$
By Lemma~\ref{lemma:spacevstime}, it follows that
$$
\size{\proctwo}\leq\nattwo\cdot\size{\procone}\leq\polytwo_{\bde{\procone}}(\size{\procone})\cdot\size{\procone}\leq\polyone_{\bde{\procone}}(\size{\procone}).
$$
This concludes the proof.
\end{proof}
}{}
Let us now consider Theorem~\ref{theo:boundint}: how can we deduce it from Proposition~\ref{prop:polysound}?
Everything boils down to show that for normal processes, the box-depth can be read off from their type.
In the following lemma, $\bde{\typone}$ and $\bde{\conone}$ are the nesting depths of $\bang$ inside the type 
$\typone$ and inside the types appearing in $\conone$ (for every type $\typone$ and context $\conone$).
\begin{lemma}
Suppose that $\tyg{\conone}{\contwo}{\conthree}{\ptone}{\cone:\typone}$ and that $\ptone$ is 
normal. Then $\bde{\widehat{\ptone}}=\max\{\bde{\conone},\bde{\contwo},\bde{\conthree},$ $\bde{\typone}\}$.
\end{lemma}
\begin{proof}
An easy induction on $\ptone$.
\end{proof}
\section{Conclusions}
In this paper, we introduced a variation on Caires and Pfenning's \pidill, called \pidsll,
being inspired by Lafont's soft linear logic. The key feature of \pidsll\ is the fact
that the amount of interaction induced by allowing two processes to interact with
each other is bounded by a polynomial whose degree can be ``read off'' from the
type of the session channel through which they communicate.

What we consider the main achievement of this paper is definitely \emph{not} the proof
of these polynomial bounds, which can be obtained by adapting the ones 
in~\cite{DalLago10a} or in~\cite{DalLago10b}, although this anyway presents
some technical difficulties due to the low-level nature of the $\pi$-calculus compared to the lambda
calculus or to higher-order $\pi$-calculus. Instead, what we found very interesting
is that the operational properties induced by typability in \pidsll, bounded interaction
\emph{in primis}, are not only very interesting and useful in practice, but different from
the ones  obtained in soft lambda calculi: in the latter, it's the \emph{normalization} time
which is bounded, while here it's the \emph{interaction} time. Another aspect that we find interesting
is the following: it seems that the constraints on processes induced by the adoption of
the more stringent typing discipline \pidsll, as opposed to \pidill, are quite natural
and do not rule out too many interesting examples. In particular, the way sessions
can be defined remains essentially untouched: what changes is the way sessions can be
offered, i.e. the discipline governing the offering of multiple sessions by servers.
All the examples in~\cite{Caires10} and the one from Section~\ref{sect:pidillia}
are indeed typable in \pidsll.

Topics for future work include the accommodation of recursive types into \pidsll. This
could be easier than expected, due to the robustness of light logics to the presence
of recursive types~\cite{DalLago06}.

\bibliographystyle{eptcs}
{\small \bibliography{main}}
\end{document}